\documentclass[lettersize,journal]{IEEEtran}
\usepackage{amssymb}
\usepackage{amscd}
\usepackage{amsmath}
\usepackage{epsfig}
\usepackage{graphics}
\graphicspath{ {./Figures/} }
\usepackage{psfrag}
\usepackage{rotating}
\usepackage{amsmath}
\usepackage{amsfonts}
\usepackage{bbm}
\usepackage{url}
\usepackage{color}
\usepackage[caption=false,font=footnotesize]{subfig}
\usepackage{epstopdf}
\usepackage{amsthm}
\usepackage{tikz}
\usetikzlibrary{patterns}
\usepackage{tkz-euclide}
\usepackage{mathtools}
\usepackage[linesnumbered,lined,boxed]{algorithm2e}
\SetArgSty{textnormal}

\usepackage[normalem]{ulem}

\usepackage{siunitx}
\DeclareSIUnit{\belmilliwatt}{Bm}

\usepackage[final]{pdfpages}
\usepackage[inline, shortlabels]{enumitem}
\usepackage{siunitx}
\DeclareSIUnit{\belmilliwatt}{Bm}
\DeclareSIUnit{\belsquaremeter}{Bsm}

\newtheorem{theorem}{Theorem}
\newtheorem{lemma}[theorem]{Lemma}

\newcommand{\Na}{N_{\rm a}}

\newcommand{\Ptx}{P_{\rm tx}}

\usepackage{multicol}
\usepackage[nolist, printonlyreused]{acronym}

\newcommand{\uc}{\mathcal{S}^1} 
\newcommand{\boreluc}{\mathcal{B}^1} 
\newcommand{\dif}{\,\text{d}} 

\newfont{\bbb}{msbm10 scaled 700}

\newfont{\bb}{msbm10 scaled 1100}
\newcommand{\CC}{\mbox{\bb C}}
\newcommand{\PP}{\mbox{\bb P}}

\newcommand{\ZZ}{\mbox{\bb Z}}

\newcommand{\HH}{\mbox{\bb H}}

\newcommand{\hh}{\mathbbm{h}}

\newcommand{\av}{{\bf a}}

\newcommand{\cv}{{\bf c}}

\newcommand{\sv}{{\bf s}}

\newcommand{\uv}{{\bf u}}
\newcommand{\wv}{{\bf w}}
\newcommand{\vv}{{\bf v}}
\newcommand{\xv}{{\bf x}}
\newcommand{\yv}{{\bf y}}
\newcommand{\zv}{{\bf z}}
\newcommand{\zerov}{{\bf 0}}
\newcommand{\onev}{{\bf 1}}

\newcommand{\Am}{{\bf A}}

\newcommand{\Cm}{{\bf C}}

\newcommand{\Fm}{{\bf F}}

\newcommand{\Hm}{{\bf H}}
\newcommand{\Id}{{\bf I}}

\newcommand{\Sm}{{\bf S}}

\newcommand{\Ym}{{\bf Y}}
\newcommand{\Zm}{{\bf Z}}

\newcommand{\Ec}{{\cal E}}
\newcommand{\Fc}{{\cal F}}
\newcommand{\Gc}{{\cal G}}

\newcommand{\Rc}{{\cal R}}

\newcommand{\Tc}{{\cal T}}

\newcommand{\muv}{\hbox{\boldmath$\mu$}}

\newcommand{\omegav}{\hbox{\boldmath$\omega$}}

\newcommand{\jsf}{{\sf j}}

\newcommand{\psf}{{\sf p}}

\newcommand{\diag}{{\hbox{diag}}}

\renewcommand{\vec}{{\rm vec}}

\usepackage{stmaryrd} 

\usepackage{hyperref}
\hypersetup{
    bookmarks=true,         
    unicode=false,          
    pdftoolbar=true,        
    pdfmenubar=true,        
    pdffitwindow=false,     
    pdfstartview={FitH},    
    pdfnewwindow=true,      
    colorlinks=true,       
    linkcolor=red,          
    citecolor=green,        
    filecolor=blue,      
    urlcolor=blue           
}

\usepackage{cleveref}
\newcommand{\T}{{\scriptscriptstyle\mathsf{T}}}
\renewcommand{\H}{{\scriptscriptstyle\mathsf{H}}}

\newcommand{\jim}{\jsf}

\newcommand\copyrighttext{%
  \footnotesize \textcopyright 2024 IEEE. Personal use of this material is permitted. Permission from IEEE must be obtained for all other uses, in any current or future media, including reprinting/republishing this material for advertising or promotional purposes, creating new collective works, for resale or redistribution to servers or lists, or reuse of any copyrighted component of this work in other works.
}
\newcommand\copyrightnotice{%
\begin{tikzpicture}[remember picture,overlay]
\node[anchor=south,yshift=10pt] at (current page.south) {\fbox{\parbox{\dimexpr\textwidth-\fboxsep-\fboxrule\relax}{\copyrighttext}}};
\end{tikzpicture}%
}

\begin{document}

    \begin{acronym}
        \acro{BS}{base station}
        \acro{UE}{user equipment}
        \acro{BT}{beam tracking}
        \acro{mmWave}{millimeter wave}
        \acro{LOS}{line of sight}
        \acro{NLOS}{non-LOS}
        \acro{AoA}{angle of arrival}
        \acro{AoD}{angle of departure}
        \acro{ULA}{uniform linear array}
        \acro{OFDM}{orthogonal frequency division multiplexing}
        \acro{HMM}{hidden Markov model}
        \acro{AWGN}{additive white Gaussian noise}
        \acro{CDF}{cummulative distribution function}
        \acro{ISAC}{integrated sensing and communication}
        \acro{RF}{radio frequency}
        \acro{MIMO}{multiple-input multiple-output}
        \acro{CP}{cyclic prefix}
        \acro{DL}{downlink}
        \acro{FD}{fully digital}
        \acro{FC}{fully connected}
        \acro{HDA}{hybrid digital analog}
        \acro{OSPS}{one-stream-per-subarray}
        \acro{TDMA}{time division multiple access}
        \acro{TDM}{time division multiplexing}
        \acro{DPSS}{discrete prolate spheroidal sequence}
        \acro{SNR}{signal to noise ratio}
        \acro{SINR}{signal to interference and noise ratio}
        \acro{HMM}{hidden Markov model}
        \acro{CDF}{cumulative distribution function}
        \acro{CSI}{channel state information}
        \acro{V2I}{vehicle-to-infrastructure}
        \acro{MU-MIMO}{multi-user MIMO}
        \acro{ADC}{analog-to-digital converter}
        \acro{CFO}{carrier frequency offset}
        \acro{DFT}{discrete Fourier transform}
        \acro{MMV}{multiple measurement vector}
        \acro{BA}{beam alignment}
        \acro{ISI}{intersymbol interference}
        \acro{ICI}{intercarrier interference}
        \acro{DFT}{discrete Fourier transform}
        \acro{IDFT}{inverse DFT}
        \acro{D2D}{device-to-device}
        \acro{ISTA}{iterative soft-thresholding algorithm}
        \acro{CS}{compressed sensing}
        \acro{MFB}{matched filter bound}
        \acro{FDE}{frequency domain equalization}
        \acro{ZF}{zero forcing}
        \acro{MMSE}{minimum mean square error}
        \acro{PAPR}{peak to average power ratio}
        \acro{MAE}{mean absolute error}
    \end{acronym}

    \author{Fernando Pedraza, Jan Christian Hauffen, Fabian Jaensch, Shuangyang Li, and Giuseppe Caire\\ 
    Communications and Information Theory, Technical University of Berlin, Germany \\
    Emails: \{f.pedrazanieto, j.hauffen, f.jaensch, shuangyang.li, caire\}@tu-berlin.de
    }

    \title{Distributed Beam Alignment in sub-THz D2D Networks}
    \maketitle
    \copyrightnotice
    
    \begin{abstract}
        Devices in a \ac{D2D} network operating in sub-THz frequencies require knowledge of the spatial channel that connects them to their peers. Acquiring such high dimensional channel state information entails large overhead, which drastically increases with the number of network devices. In this paper, we propose an accelerated method to achieve network-wide beam alignment in an efficient way. To this aim, we consider \ac{CS} estimation enabled by a novel design of pilot sequences. Our designed pilots have constant envelope to alleviate hardware requirements at the transmitters, while they exhibit a ``comb-like'' spectrum that flexibly allocates energy only on certain frequencies. This design enables multiple devices to transmit thier pilots concurrently while remaining orthogonal in frequency, achieving simultaneous alignment of multiple devices. Furthermore, we present a sequential partitioning strategy into transmitters and receivers that results in logarithmic scaling of the overhead with the number of devices, as opposed to the conventional linear scaling. Finally, we show via accurate modeling of the indoor propagation environment and ray tracing simulations that the resulting sub-THz channels after successful beamforming are approximately frequency flat, therefore suitable for efficient single carrier transmission without equalization. We compare our results against an "802.11ad inspired" baseline and show that our method is capable to greatly reduce the number of pilots required to achieve network-wide alignment. 
    \end{abstract}


    \section{Introduction}

The explosive growth of wireless communication applications has driven the demand for higher data rates and greater network capacity. To meet these requirements, \ac{mmWave} and sub-THz frequencies, particularly within the 30-300 \si{\giga\hertz} range, have emerged as a promising solution for indoor applications that require extremely high rates (e.g., supporting VR gargles, 8K TV sets, and in some specific industrial campus applications) \cite{ThzCommSurvey}. Specifically, the IEEE 802.11ad standard \cite{80211adAmendment}, operating at \SI{60}{\giga\hertz}, demonstrates the potential of high frequency communication systems to achieve multi-gigabit per second data rates in indoor local area networks.

However, the deployment of sub-THz communication introduces significant challenges due to the inherent propagation characteristics of high-frequency signals \cite{Maltsev2010statistical}. At these frequencies, the propagation losses are substantially higher compared to those at the lower microwave range. This necessitates the use of highly directional antennas and beamforming techniques to compensate for path loss and maintain the link budget. Beamforming not only helps in overcoming high propagation losses but also plays a critical role in minimizing interference in dense \ac{D2D} networks \cite{GameBasedBeamScheduling}, where spatial reuse is essential for efficient spectrum utilization.

In this context, \ac{BA} emerges as a crucial process in sub-THz networks \cite{InitialAccessmmWaveTHz}. Efficient \ac{BA} ensures that the transmitter and receiver beams are accurately aligned, maximizing the signal strength and minimizing the probability of link failure. However, the \ac{BA} process introduces additional overhead, which can negatively impact the  overall performance of the network, especially in dynamic indoor environments where devices may frequently change their positions. Therefore, it is imperative that the \ac{BA} procedure is not only accurate but also fast, to reduce the overhead and maintain high data rates. The sparse nature of sub-THz channels, which are often composed of a dominant \ac{LOS} path and a limited number of strong reflections, offers an opportunity to optimize this alignment process using \ac{CS} techniques \cite{CompressedSensing}. \ac{CS} leverages the sparse nature of the sub-THz channel, allowing the \ac{BA} process to be performed more efficiently by reducing the number of measurements needed to identify the optimal beam directions. This approach is particularly beneficial in scenarios where rapid alignment is necessary, such as in mobile or dense indoor environments, where frequent realignments may be required due to device mobility or environmental changes. 

\ac{CS} techniques have garnered considerable attention for their ability to reduce the overhead associated with \ac{BA} in \ac{mmWave} and sub-THz systems. The seminal work \cite{HeathAlkhateeb} proposed a \ac{CS}-based approach that capitalizes on the sparsity of the \ac{mmWave} channel. Their method effectively reduces the number of measurements needed to identify the optimal beam directions, thereby minimizing the overhead of the \ac{BA} process. Building on this foundation, \cite{PeakFindingBA} introduced another \ac{CS}-based method for initial access in high frequency systems using analog arrays. Their approach employs gradient-type methods and relies on a codebook-based beamforming strategy. The authors argue that designing precoders and combiners on-the-fly would be much slower compared to a codebook-based approach, which motivates us to consider codebook-based beamforming in our work. Their method hinges on the assumption that the channel gain after beamforming is, at least locally, a monomodal function of the beam index (i.e. the channel gain after beamforming always decreases when the difference between the index of a beam and the index of the optimal beam is made larger). However, this assumption may not hold in indoor scenarios, where a larger number of reflections can disrupt the monotonic relationship, leading to potential performance degradation in such environments. 

Beyond the theoretical work on \ac{CS}-based \ac{BA}, several works have experimentally demonstrated their practical feasibility. In particular, the authors of \cite{CompressiveBAExperimentalSetup} implemented their proposed method on an off-the-shelf 802.11ad compliant router with 32 antenna elements, achieving successful performance. Similarly, \cite{CNN_BA_P2P} proposes a convolutional neural network based algorithm and test it on different hardware testbeds at the \SI{60}{\giga\hertz} band.

Most works in the high frequency \ac{D2D} literature assume partial \cite{D2DBeamAlignmentError, SpatioTempD2DMisalignment} of full \cite{D2DmmWaveAndMicrowave, D2DCoalitionFormation, GBLinks} \ac{CSI} knowledge is available to a centralized controller. Then, an approximate solution to the problem of optimal resource allocation with directional links is searched for, by means of graph neural networks \cite{GBLinks} and other heuristic approaches \cite{D2DCoalitionFormation, D2DmmWaveAndMicrowave}. The information assumed known to the controller consists of the  instantaneous channel matrix (often of large dimension when considering sub-THz systems) between every pair of devices. In practice, due to the large number of parameters to be estimated and the short timescale of variation of small scale fading, the \ac{CSI} might be completely obsolete by the time it is available to the controller. Another widely researched topic analyzes the effect of partial or imprecise knowledge \cite{D2DBeamAlignmentError, SpatioTempD2DMisalignment}, and demonstrates the detrimental effects of suboptimal beam pointing at the network level. While all those works are relevant for sub-THz \ac{D2D} network research, we specifically focus on efficient estimation of \ac{CSI} that forms the base for resource management in upper layers. Furthermore, our efforts are aimed at estimating the long-term statistics of the channels between every pair of devices, which remain constant in a much larger time scale under the assumption of limited mobility. Other research activities on directional \ac{D2D} networks have focused on dynamic beam tracking \cite{LiSteer, MultiLobeD2DTracking, ABPD2D, UAV_D2D_Tracking}, in order to reduce the frequency of full alignment rounds. Solutions consist of nonstandalone systems using light sensors \cite{LiSteer}, variations of classical monopulse radar tracking \cite{MultiLobeD2DTracking}, or heuristic approaches that periodically explore neighboring directions to detect when beam adaptation is necessary \cite{ABPD2D, UAV_D2D_Tracking}. However, all those works assume initial knowledge of the relevant \ac{CSI} is available and aim to dynamically update the estimates to prevent the loss of performance caused by beam misalignments and to minimize channel estimation overhead. We remark once again that our work deals with the acquisition of network-wide \ac{CSI} which is required as a starting point for the aforementioned works.

The most similar work to ours can be found in \cite{HanzoCoordtinatedMACD2D}, where the authors consider a similar setup, but focus on a  higher layer approach as opposed to our physical layer methods. 
Furthermore, they differentiate between access points, which are connected to each other through wired links, and users. Instead, we focus on purely \ac{D2D} scenarios without any infrastructure. To the best of our knowledge, the problem of network-wide \ac{BA} in a \ac{D2D} network has not been investigated yet. 

This paper describes a new method based on \ac{CS} techniques and a novel pilot design for fast beam alignment in sub-THz \ac{D2D} networks assuming no side information. Specifically, our main contributions can be summarized as follows:

\begin{itemize}
    \item We derive a model for the received signal model that allows us to cast \ac{BA} as an instance of the \ac{MMV} problem in \ac{CS}, which can then be solved using conventional techniques.
    \item We propose a pilot signal design that flexibly spreads the energy among a desired number of discrete frequencies while keeping a constant envelope in time domain, thereby remaining energy efficient. This pilot design is independent of any processing in the spatial domain and can therefore be applied to any scenario, not only sub-THz networks with beamforming capable devices, thus having independent merit.
    \item Based on the proposed pilot design, we propose an accelerated method for network-wide \ac{BA} where multiple devices simultaneously transmit on orthogonal frequency sets. In this way, the number of pilot transmissions reduces from $K$ to $\lceil\log_{2}K\rceil$, where $K$ indicates the total number of devices in the network.
    \item We study the feasibility of single carrier communication after beam alignment by evaluating the achievable performance under different receiver conditions. 
    \item Our unified framework includes, as a special case, also some forms of beam sweeping closer, in the operation principle, to 802.11ad. We take such scheme as the baseline against which to compare the more advanced \ac{CS}-based scheme, and show important performance gains.
\end{itemize}

\begin{figure*}
    \centering
    \subfloat{\resizebox{!}{0.22\textheight}{\input{Figures/D2DNetworkExample}}}
    \qquad
    \subfloat{\includegraphics[width=0.4\linewidth]{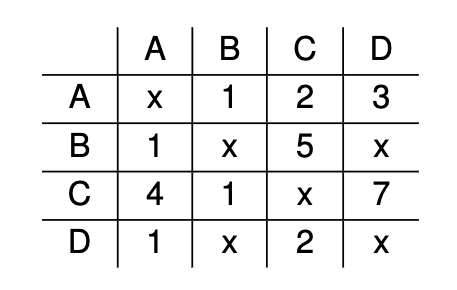}}
    \caption{Example of D2D network and final objective. The left side of the picture shows a network with four devices and some example beams from their respective codebooks, labeled by their indices. On the right, the desired end result, where each device is aware of the beam index needed to communicate with every device in the network.}
    \label{fig:D2D_net_and_table}
\end{figure*}

\subsection{Paper organization and notation}
The rest of this paper is organized as follows: Section \ref{sec:sys_model} describes the signal model in time and frequency domain. We derive the \ac{MMV} formulation in Section \ref{sec:SingleBA}, and we describe the design of pilot sequences in Section \ref{sssec:Seq_gen}. Section \ref{sec:NetworkBA} extends the methods of Section \ref{sec:SingleBA} to a network scenario to achieve logarithmic scaling of the total beam alignment time with the number of devices. Section \ref{sec:CommPhase} describes how to apply the results of \ac{BA} to communicate data between pairs of devices, and studies the feasibility of single carrier transmission without equalization. The proposed methods are evaluated and their performance presented in Section \ref{sec:Results}, where we also proposed an 802.11ad inspired baseline. Section \ref{sec:Conclusions} concludes the work and discusses potential future work topics.

Throughout the document, we consider the following notation. We denote scalars, matrices and vectors with $a$, $\Am$ and $\av$ respectively. $\Am^\T$, $\Am^\H$, $\Am^{\rm c}$, $[\Am]_{i, j}$ and $\Am_{i, :}$ indicate respectively the transpose, Hermitian transpose, element-wise conjugate, the $(i, j)$-th element, and the $i$-th row. $\|\av\|_{2}$ and $[\av]_{i}$ are respectively the $\ell_{2}$ norm of vector $\av$ and its $i$-th entry. Symbols $\hh$ and $\HH$ are reserved for the flattened version of a matrix and a matrix whose columns are flattened matrices. For a signal $\xv$ in discrete-time domain, we indicate its \ac{DFT} by $\tilde{\xv}$, and for a beamforming vector $\uv$ we indicate its transformation to a beamspace domain by $\check{\uv}$. We abbreviate the index set $\{1, \dots, M\}$ as $[M]$. The imaginary unit is denoted as $\jim \coloneqq \sqrt{-1}$. $\onev_{\Na}$ is the $\Na$ dimensional vector of all ones.
    \section{System Model}\label{sec:sys_model}

We consider an indoor local network scenario where $K$ multi-antenna devices wish to communicate with each other using directional beams. In particular, for simplicity, we let all devices be equipped with a \ac{ULA} of $N_{\rm a}$ antennas driven by a single \ac{RF} chain. Prior to any data communication, each device needs to know the beam configuration to use in order to interact with any other node in the network. Focusing on devices with beam codebooks of finite cardinality, the goal is to achieve \textit{network beam alignment} as exemplified in Fig.~\ref{fig:D2D_net_and_table}, where each device is aware of the beam index needed to establish any communication link. 
We begin by presenting the channel, signal model, and receive processing for any arbitrary pair of devices, and defer the discussion on network-wide operation to Section \ref{sec:NetworkBA} for the sake of clarity of exposition.

\subsection{Channel Model}

Given the slowly varying nature of the indoor channels considered, we approximate their impulse response as locally time-invariant. Within this setup, we describe the channel impulse response between the $j$-th antenna at the transmitter and the $i$-th antenna at the receiver by
\begin{align}\label{eq:Channel_NF}
    [\Hm(\tau)]_{i, j} = \sum_{p=0}^{P}\beta_{p}\delta(\tau - \tau^{i, j}_{p}),
\end{align}
where $i, j\in [N_{\rm a}]$, $p=0$ indicates the \ac{LOS} path and all other terms in the sum are \ac{NLOS}, and $\tau^{i, j}_{p}$ indicates the propagation delay between the $j$-th and $t$-th antenna elements at the transmitter and the receiver respectively when traversing the $p$-th path and $\beta_{p}$ is the associated path gain. For this work, we directly obtain values of $\tau^{i, j}_{p}$ and $\beta_{p}$ from the state of the art ray tracing simulator Wireless InSite \cite{InSite}\footnote{Notice that, since the channel is given in terms of paths and delays from antenna element to antenna element, no specific far field assumption between the devices arrays is made, and our model automatically includes far field or near field conditions, depending on the relative position of the devices.}.

\subsection{Signal Model}\label{ssec:sig_model}
Consider a sequence of symbols $\{s[n]\}_{n=0}^{M-1}$, to be transmitted at symbol rate $T$ by means of a pulse shaping filter $g_{\rm tx}(t)$, resulting in the time domain signal
\begin{align}
    s(t) = \sum_{n=0}^{M-1}s[n]g_{\rm tx}(t-nT),
\end{align}
where the pulse shaping filter has unit energy ($\int_{-\infty}^{\infty}|g_{\rm tx}(t)|^{2}\,dt = 1$) and the sequence $\{s[n]\}_{n=0}^{M-1}$ satisfies $\sum_{n=0}^{M-1}|s[n]|^{2} = \Ec$, where the term $\Ec$ indicates the total energy of the sequence. The associated average transmit power is given by $P_{\rm tx} \coloneqq \Ec/(MT)$.

If the transmitter applies a (unit norm) beamforming vector $\vv\in\CC^{N_{\rm a}}$, the signal at the receiving device is given by
\begin{align}
    \yv(t) &= (\Hm \ast \vv s)(t) + \zv(t)
\end{align}
where $\ast$ indicates linear convolution and $\zv(t)$ is an \ac{AWGN} process. The receiver applies a (unit norm) combining vector $\uv\in\CC^{N_{\rm a}}$ and a filter $g_{\rm rx}(t)$ before sampling the signal at rate $1/T$ \footnote{In this work, we assume that sampling frequency and carrier frequency offsets have been previously estimated and compensated with conventional methods such as the one described in \cite{CFO_sync_classic} or in case of multiple devices the distributed method in \cite{TO_CFO_DiMIMO}.} with the first sample being taken at $t=\tau_{0}$. The sampled signal is thus given by 

\begin{align}
    y[n'] &= \left[g_{\rm rx} \ast \uv^\H \yv \right](t)\Big|_{t=n'T+\tau_{0}}\nonumber\\
    &= \uv^\H\left[g_{\rm rx} \ast \Hm \ast \vv s\right](t) + \uv^\H\zv'(t)\Big|_{t=n'T+\tau_{0}}\nonumber\\
    &= \sum_{n=0}^{M-1}s[n]\uv^\H\left[g_{\rm rx} \ast \Hm \ast g_{\rm tx}\right](t)\vv\Big|_{t=n'T+\tau_{0}} + z'[n'] \nonumber\\
    &= \sum_{n=0}^{M-1}s[n]\uv^\H\overline{\Hm}[n'-n]\vv+ z'[n']\label{eq:DT_IO},
\end{align}
where $\overline{\Hm}[n'] \coloneqq [g_{\rm rx} \ast \Hm \ast g_{\rm rx}](t)\Big|_{t=n'T + \tau_{0}}$ is the discrete-time equivalent channel capturing pulse shaping and timing offsets and $z'[n']$ is a discrete time \ac{AWGN} process. Notice that $[\overline{\Hm}[n']]_{i, j} \neq 0$ only if $n'T + \tau_{0} \approx \tau^{i, j}_{p}$ for some path $p$. We define the channel length $L$ as $L \coloneqq \min\{n \mid \overline{\Hm}[n'] = 0\; \forall n'\geq n\}$\footnote{In practice, in a network with many devices, $L$ should be the maximum of the channel lengths between any pair of devices.}. For a given channel length, letting the receiver take $M+L-1$ samples, \eqref{eq:DT_IO} can be rewritten in vector form as
\begin{align}\label{eq:TD_IO_vec}
    \yv = \Sm[\overline{\hh}[0], \dots, \overline{\hh}[L-1]]^\T \wv + \zv',
\end{align} 
where $\overline{\hh}[n] \coloneqq \vec(\overline{\Hm}[n]) \in \CC^{N_{\rm a}^2}$, and we defined the combined transmit-receive beamforming vector $\wv \coloneqq \vv \otimes \uv^{\rm c} \in \CC^{N_{\rm a}^2}$, and $\Sm\in\CC^{(M+L-1)\times L}$ is the Toeplitz matrix
\begin{align}\label{eq:S_k}
\Sm \coloneqq 
    \begin{bmatrix}
        s[0] & 0 & \dots & 0 \\
        s[1] & s[0] & \dots & \vdots \\
        \vdots & \vdots &  & s[0]\\
        s[M-1] & s[M-2] & \ddots & \vdots\\
        0 & s[M-1] & & \vdots\\
        \vdots & \vdots & \ddots & \vdots \\
        0 & 0 & \dots & s[M-1]
    \end{bmatrix}.
\end{align}

\subsection{Equivalent Frequency Domain Model}
The channel model \eqref{eq:S_k} allows us to operate with time domain sequences, which is useful at the transmitter e.g. to ensure the transmitted signal has low \ac{PAPR} and thereby alleviate hardware requirements in end devices. However, a frequency domain model facilitates receive processing, by transforming a time dispersive channel into a set of parallel independent channels, as it is known from \ac{OFDM}.
To this aim, we need to prepend the sequence $s[n]$ with a \ac{CP} as done in the uplink of LTE (see e.g. \cite{LTE_Uplink}). For a \ac{CP} length $L' \geq L$, and assuming the propagation delay of the \ac{LOS} path is exactly $\tau_{0}$\footnote{This is just done for notation simplicity. As long as we take a window of $M$ samples starting anywhere from $\min_{i,j} \tau_{0}^{i, j} + LT$ and $\min_{i,j} \tau_{0}^{i, j} + L'T$, the following derivations hold up to a cyclic rotation, which would only introduce some irrelevant phase factors in the subsequent analysis.} the $M$ samples corresponding to elements $n'=L', \dots L'+M-1$ of $\yv$ in \eqref{eq:TD_IO_vec} are given by
\begin{align}\label{eq:TD_IO_vec_CP}
    \mathring{\yv} = \mathring{\Sm}[\overline{\hh}[0], \dots, \overline{\hh}[L-1], \zerov, \dots, \zerov]^\T \wv + \zv',
\end{align}
where $\mathring{\Sm} \in \CC^{M\times M}$ is the circulant matrix
\begin{align}
    \mathring{\Sm} \coloneqq 
    \begin{bmatrix}
        s[0] & s[M-1]& \dots& s[1]\\
        s[1] & s[0]& \dots& s[2]\\
        \vdots & \vdots & \ddots & \vdots \\
        s[M-1] & s[M-2] & \dots &s[0]
    \end{bmatrix},
\end{align}
which can be diagonalized by a \ac{DFT} basis as
\begin{align}
    \mathring{\Sm} = \frac{1}{M}\Fm_{M}^\H \diag(\tilde{s}[0], \dots, \tilde{s}[M-1])\Fm_{M},
\end{align}
where $\Fm_{M}$ is the $M\times M$ \ac{DFT} matrix with entries $[\Fm]_{i, j} = e^{-\jim\frac{2\pi}{M}(i-1)(j-1)}$ for $i, j \in [M]$ satisfying $\frac{1}{M} \Fm_{M}\Fm_{M}^\H = \Id_{M}$, where $\Id_{M}$ is the identity matrix of rank $M$, and $\tilde{s}[m]$ is the $m$-th element of the $M$-point \ac{DFT} of $[s[0], \dots, s[M-1]]$. 
By applying a (normalized) \ac{DFT} to the measurement \eqref{eq:TD_IO_vec}, we obtain
\begin{align}
    \tilde{\yv} &= \frac{1}{\sqrt{M}}\Fm_{M}\mathring{\yv} \nonumber\\
    &=\frac{1}{\sqrt{M}}\Fm_{M}\mathring{\Sm}[\overline{\hh}[0], \dots, \overline{\hh}[L-1], \zerov, \dots, \zerov]^\T \wv + \tilde{\zv}' \nonumber\\
    &=\diag\left(\frac{\tilde{s}[0]}{\sqrt{M}}, \dots, \frac{\tilde{s}[M-1]}{\sqrt{M}}\right)\tilde{\HH}^\T \wv + \tilde{\zv}'\label{eq:FD_IO},
\end{align}
where we defined $\tilde{\HH}^\T$ as $[\tilde{\overline{\hh}}[0], \dots, \tilde{\overline{\hh}}[M-1]] \in \CC^{N_{\rm a}^{2}\times M}$ as the frequency domain channel matrix, with $\tilde{\overline{\hh}}[m] \in \CC^{N_{\rm a}^2}$ indicating the vectorized channel at frequency index $m$, and where $\tilde{\zv}'$ has the same statistics as $\zv'$ since the transformation was unitary.

    \section{Single device beam alignment}\label{sec:SingleBA}

We describe here a procedure to find the best beams in a codebook connecting a transmitting and a receiving device. In Section \ref{sec:NetworkBA}, we describe how to extend the method here to allow for multiple devices to simultaneously align.

First, we derive an equivalent version of the measurement model in \eqref{eq:FD_IO} in the transformed domain induced by the beam codebooks. For simplicity, we assume that all devices use the same codebook $\Cm = [\cv_{1}, \dots, \cv_{N_{\rm a}}]$, given by the normalized \ac{DFT} matrix of dimension $N_{\rm a}$, i.e. $\Cm = 1/\sqrt{N_{\rm a}}\Fm_{N_{\rm a}}$\footnote{We focus on a specific configuration to simplify exposition, but all the derivations hold for any codebook with approximately orthogonal columns on which channels have a sparse representation.}. We assume that a device can generate a beam by linearly combining entries of the codebook, i.e. 
\begin{align}
    \vv = \Cm\check{\vv},
\end{align} 
where we refer to $\check{\vv}$ as the \textit{beam weighting vector}. A specific design for the beam weighting vectors is presented at the end of this section. The combined transmit-receive beamforming vector introduced in Section \ref{ssec:sig_model} can then be rewritten as
\begin{align}
    \wv = \vv \otimes \uv^{\rm c} = (\Cm\check{\vv}) \otimes (\Cm\check{\uv})^{\rm c}.
\end{align}
The $m$-th entry of $\tilde{\yv}_{r, k}$ in in \eqref{eq:FD_IO}, corresponding to the $m$-th frequency index, is then given by
\begin{align}
    [\tilde{\yv}]_{m} &= \frac{\tilde{s}[m]}{\sqrt{M}}\vec(\tilde{\Hm}[m])^\T[(\Cm\check{\vv}) \otimes (\Cm^{\rm c}\check{\uv}^{\rm c})] + [\tilde{\zv}']_{m}\nonumber\\
    &=\frac{\tilde{s}[m]}{\sqrt{M}}[(\Cm\otimes \Cm^{\rm c})(\check{\vv}\otimes\check{\uv}^{\rm c})]^\T\vec(\tilde{\Hm}[m]) + [\zv']_{m}\nonumber\\
    &=\frac{\tilde{s}[m]}{\sqrt{M}} (\check{\vv}^\T \otimes \check{\uv}^\H) \vec(\Cm^\H \tilde{\Hm}[m] \Cm) + [\zv']_{m}\nonumber\\
    &= \frac{\tilde{s}[m]}{\sqrt{M}}(\check{\vv}^\T \otimes \check{\uv}^\H) \check{\hh}[m] + [\zv']_{m},
\end{align}
where we defined $\check{\hh}[m] \coloneqq \vec(\Cm^\H \tilde{\Hm}[m] \Cm)$, which represents the vectorized channel matrix at frequency index $m$ on the beamspace induced by codebook $\Cm$. The complete measurement in \eqref{eq:FD_IO} can then be expressed as 
\begin{align}\label{eq:FD_IO_BS}
    \tilde{\yv}^\T &= \frac{1}{\sqrt{M}}(\check{\vv}^\T \otimes \check{\uv}^\H)[\check{\hh}[0], \dots, \check{\hh}[M-1]]\diag(\tilde{\sv}) + \zv'^\T \nonumber\\
    &=\frac{1}{\sqrt{M}}(\check{\vv}^\T \otimes \check{\uv}^\H)\check{\HH}\diag(\tilde{\sv}) + \zv'^\T,
\end{align}
where $\tilde{\sv} \coloneqq [\tilde{s}[0], \dots, \tilde{s}[M-1]] \in \CC^{M}$ and we defined $\check{\HH} \coloneqq [\check{\hh}[0], \dots, \check{\hh}[M-1]]$. Notice that each row of $\check{\HH}$ corresponds to the frequency response of the channel resulting from a transmit-receive beam pair combination. For the sub-THz channels considered here, the matrix $\check{\HH}$ presents approximate row sparsity, i.e. only a few rows of $\check{\HH}$ will carry significant gain, corresponding to pairs of beams that align the transmitter and the receiver. Therefore, we can now formulate the beam alignment problem as the estimation of the set of indices corresponding to the rows with largest norm, referred to as the row support of $\check{\HH}$.

In order to do this, we now collect $Q$ measurements as in \eqref{eq:FD_IO_BS}, so we introduce a new index $q$ representing each pilot sequence transmission, where the transmitter repeats the same sequence $\{s[n]\}$, whereas the transmit and receive beamforming vectors at the $q$-th measurement are now given by $\vv_{q}$ and $\uv_{q}$ respectively. As in \cite{Xiaoshen2018scalable,Xiaoshen2019efficient}, it is assumed that the receiver is aware of the sequence of beam weighting vectors used by the transmitter, e.g. as part of the specification of a standard. The complete measurement block of dimension $M\times Q$ is given by 
\begin{align}\label{eq:MMV}
    \Ym = \frac{1}{\Na\sqrt{M}}\Am\check{\HH}\diag(\tilde{\sv}) + \Zm = \frac{1}{\Na\sqrt{M}}\Am\check{\HH}' + \Zm,
\end{align}
where $\Am \coloneqq \Na[(\hat{\vv}_{1}^\T\otimes \hat{\uv}_{1}^\H)^\T, \dots, (\hat{\vv}_{Q}^\T \otimes \hat{\uv}_{Q}^\H)^\T]^\T \in \CC^{Q\times N_{\rm a}^2}$. Notice that $\check{\HH}' \coloneqq \check{\HH}\diag(\tilde{\sv})$ preserves the row sparse structure of $\check{\HH}$.

Given the structure of the measurement model \eqref{eq:MMV}, we can apply techniques from the \ac{CS} literature \cite{SparSA} to the recovery of the row support of $\check{\HH}'$. Compressed sensing techniques allow us to solve the problem with a number of pilots $Q \ll N_{\rm a}^{2}$. In particular, we can formulate the \ac{MMV} problem
\begin{align}
    \underset{\check{\HH}'}{\rm minimize}\qquad \left\|\Ym - \frac{1}{\Na\sqrt{M}}\Am\check{\HH}'\right\|_2^2 + \gamma\|\check{\HH}'\|_{2,1}, \label{eq:cs_mmv}
\end{align}
where $\|\check{\HH}'\|_{2,1} = \sum_{m=0}^{M-1}\|\check{\hh}[m]\tilde{s}[m]\|_2$ is the $\ell_{2,1}$-norm and $\gamma$ is a regularization parameter trading off reconstruction error and sparsity of the solution. In this work, we solve problem \eqref{eq:cs_mmv} by applying block \ac{ISTA}, which is a well known proximal gradient method \cite{SparSA}. Notice that more efficient solvers for problem \eqref{eq:cs_mmv} have been proposed, but optimization of the specific solver is out of scope for this work. We remark that the estimate yields both the optimal transmit and receive beam indices, while all processing is done at the receiver.

Most of the \ac{CS} literature has considered sampling matrices $\Am$ with random i.i.d. entries sampled from different distributions. However, due to the highly structured nature of our sampling matrix, most choices of the beam weighting vectors result in a sampling matrix with non independent entries. Nevertheless, we show in Appendix \ref{ap:SamplingMatrix} that 
when we sample each entry of $\check{\vv}_{q}$ (identically $\check{\uv}_{q}$) for all $q$ uniformly from the unit circle, the entries of vector $\check{\vv}_{q}^\T \otimes \check{\uv}_{q}^\H$ are i.i.d.. Since $\check{\uv}_{q}$ and $\check{\vv}_{q}$ must have unit norm, their entries are sampled uniformly from the circle of radius $1\sqrt{\Na}$.

Interestingly, we observe that the presence of the $\diag(\tilde{\sv}_{k})$ term allows us to flexibly trade off the number of useful measurements and the \ac{SNR} per measurement. We dedicate the next section to propose a sequence design method that enables us to take advantage of this observation.

    \section{Pilot Sequence Design}\label{sssec:Seq_gen}
We begin this section by providing a high level motivation for the design presented here. For any given measurement in the form of \eqref{eq:MMV}, assume that the pilot sequence $\{s[n]\}_{n=0}^{M-1}$ is designed such that its total energy $\Ec$ is spread uniformly across the whole bandwidth, such that $|\tilde{s}[m]|^{2} = \Ec/M$ for $m=0,\dots, M-1$. Then, each of the measurements consists of $M$ samples whose \ac{SNR} is proportional to $\Ec/M$. On the other extreme, assume that the sequence is designed such that all the energy is concentrated in frequency index $m_{1}$, such that $|\tilde{s}[m]|^{2} = \Ec$ if $m=m_{1}$ and is 0 otherwise. Then, each measurement has a single relevant sample, whose \ac{SNR} is proportional to $\Ec$. Therefore, we identify a trade off between measurement dimension and \ac{SNR} per sample. Notice that this also affects the dimensions of problem \eqref{eq:cs_mmv}, such that the solver complexity is reduced when the energy is concentrated among a small subset of frequencies.

We dedicate this section to introduce a pilot sequence design that allows us to flexibly allocate energy on a subset of frequencies to study the practical implications of this trade off. To remain energy efficient, we further require the resulting pilots to have constant envelope in the time domain.
Finally, in order to exploit frequency diversity in the channel and prevent \ac{ICI} from tightly spaced frequencies under imperfect synchronization, we design the sequences such that the active frequencies are maximally spaced. The method is inspired by the works in \cite{FALP} and \cite{Unimod_seq_design}, where a similar strategy is used to concentrate energy among a small set of spatial sectors while using beamforming vectors with constant magnitude. We remark that the pilot design is independent of the beamforming strategy, and can be used in general setups including single antenna networks.

\begin{figure*}[!t]
    \centering
    \resizebox{\linewidth}{!}{
        \input{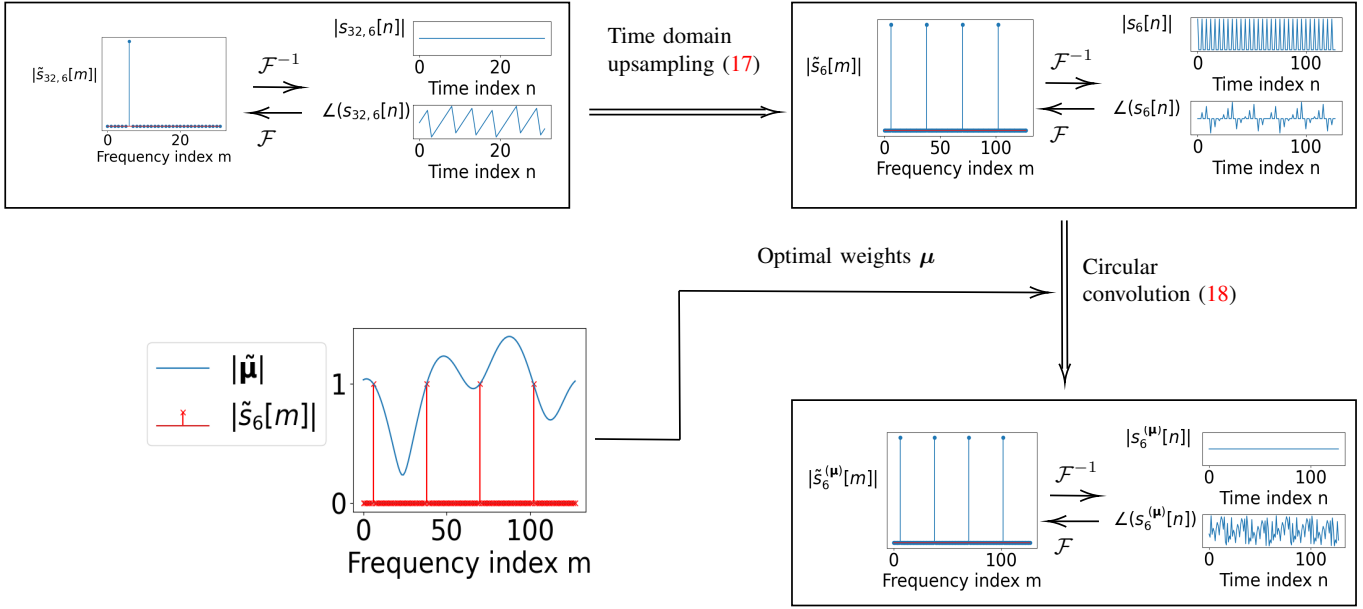}
    }
    \caption{Example of the proposed pilot sequence design method, for $M=128$, $M_{\rm s}=4$ and $k=6$. The top left subfigure shows the starting point, the $k$-th element of the \ac{DFT} basis of dimension $M_{\rm s}$. Upsampling in time domain, the spectrum is replicated and scaled as shown in the top right, but in the time domain this results in an on-off sequence, which is undesirable for energy efficient operation. Finding the optimal weights $\muv$ and convolving the upsampled signal with them, we obtain the desired pilot in the bottom right, with energy distributed only among $M_{\rm s}$ frequency indices and constant envelope in time domain.}
    \label{fig:SeqGenDiagram}
\end{figure*}

Let $M_{\rm s}$ be the desired number of frequencies among which to spread the total energy. For simplicity, we assume that $M_{\rm s}$ evenly divides the total sequence length $M$ and let $\eta \coloneqq M/M_{\rm s}$. For any finite length sequence $s[n]$, we define the circular shift operator as $[{\rm cshift}_{\eta}(s[n], k)]_{m} = s[(m-k) \mod \eta]$ for $m\in\{0, \dots, \eta -1\}$. We begin from the $\eta$-sample sequence $\tilde{s}_{\eta, k}[m] = {\rm cshift}_{\eta}([1, 0, \dots, 0], k)$ for any sequence index $k\in\{0, \dots, \eta-1 \}$. We interpret $\tilde{s}_{\eta, k}[m]$ as the frequency domain representation of the $\eta$-sample sequence $s_{\eta, k}[n] \coloneqq {\rm IDFT}(\tilde{s}_{\eta, k}[m], \eta)$, where ${\rm IDFT}(s[m], \eta)$ is the $\eta$-point \ac{IDFT} of $s[m]$. Due to the properties of the \ac{DFT} \cite{Unimod_seq_design}, we know that the spectrum of the upsampled $M$-long signal
\begin{align}\label{eq:upsampling}
    s_{k}[n] = 
    \begin{cases}
        s_{\eta, k}[n/M_{\rm s}]\qquad &{\rm if} \;\;(n\mod M_{\rm s}) = 0\\
        0 & {\rm else}
    \end{cases}
\end{align}
is a replicated and scaled version of $\tilde{s}_{\eta, k}[m]$. Therefore, $\tilde{s}_{k}[m] = {\rm DFT}(s_{k}[n], M)$ has its energy distributed among $M_{s}$ maximally spaced frequencies, but the on-off shape of the time domain sequence $s_{k}[n]$ (see Fig.~\ref{fig:SeqGenDiagram}) would result in undesirable losses at the power amplifiers. Following the same development as \cite{Unimod_seq_design}, we observe that any circular shift of a time domain sequence does not alter the magnitude of its spectrum. Consequently, for any $M_{\rm s}$-long weight vector $\muv$ the sequence
\begin{align}\label{eq:circ_conv}
    s_{k}^{(\muv)}[n] = \sum_{n'=0}^{M_{\rm s}-1} {\rm cshift}_{M}(s_{k}[n], n')[\muv]_{n'}
\end{align}
has constant amplitude and its spectral energy is concentrated among the $M_{\rm s}$ discrete equally spaced frequencies with indices in $\Fc_{k}\coloneqq\{k, k+\eta, \dots, k+(M_{\rm s}-1)\eta\}$. The only remaining part is to design a weight vector $\muv$ that optimizes some criterion. In our case, we intend for the resulting pilot sequence to have its energy distributed as uniformly as possible among the $M_{\rm s}$ active frequencies. 

If we identify \eqref{eq:circ_conv} as a circular convolution, the corresponding operation in the frequency domain becomes the element-wise product between $\tilde{s}_{k}[m]$ and $\tilde{\muv}\coloneqq {\rm DFT}(\muv, M)$, where the \ac{DFT} is performed on the zero-padded version of $\muv$. Since $\tilde{s}_{k}[m]$ has constant magnitude on $\Fc_{k}$, the goal is for $\tilde{\muv}$ to be approximately constant (in magnitude) on $\Fc_{k}$. It can be shown that $[\tilde{\muv}]_{m\in\Fc} = {\rm DFT}(\muv', M_{\rm s})$, where $[\muv']_{n} = [\muv]_{n}e^{-\jim \frac{2\pi}{M}kn}$. The final step is therefore to find an $M_{\rm s}$-long sequence whose $M_{\rm s}$-point \ac{DFT} has approximately constant magnitude. Such a sequence can be found via classical methods like the PeCAN algorithm \cite{PeCAN}. An illustration of the procedure presented here is shown in Fig.~\ref{fig:SeqGenDiagram}. Furthermore, a simplified description of the steps required is given in Algorithm \ref{alg:seq_design}

\begin{algorithm}
\SetKwInOut{Input}{Input}
\SetKwInput{Define}{Define}
\caption{Constant amplitude comb-like spectrum pilot sequence design.}\label{alg:seq_design}
\Input{Sequenc length $M$, number of active frequencies $M_{\rm s}$}
\Define{$\eta \gets M/M_{\rm s}$}
\Input{Pilot index $k\in\{0, 1, \ldots \eta - 1\}$}
\BlankLine
Construct $\tilde{s}_{\eta, k}[m]$ by cyclically shifting $k$ times the $\eta$-sample vector $[1, 0, \dots, 0]$ \;
$s_{\eta, k}[n] \gets {\rm IDFT}(\tilde{s}_{\eta, k}[m], \eta)$\;
Obtain $s_{k}[n]$ from $s_{\eta, k}[n]$ as described in \eqref{eq:upsampling}\;
Obtain an $M_{\rm s}$-sample sequence $\muv'$ satisfying $|{\rm DFT}(\muv', M_{\rm s})|$ approximately constant, using the PeCAN algorithm from \cite{PeCAN}\;
$[\muv]_{n} \gets [\muv']_{n}e^{-\jim\frac{2\pi}{M}kn}$\;
Obtain the desired sequence $s_{k}^{(\muv)}[n]$ by combining cyclic shifts of $s_{k}[n]$ with weights in $\muv$ as described in \eqref{eq:circ_conv}\;
\end{algorithm}

    \section{Network Wide Beam Alignment}\label{sec:NetworkBA}
We now apply the theory derived in the previous sections to efficiently achieve pairwise beam alignment in a \ac{D2D} network. Considering a network as exemplified in Fig.~\ref{fig:D2D_net_and_table}, a classic one-sided beam sweeping approach (as in 802.11ad \cite{80211adAmendment}) would need $K$ rounds of measurements of the type of \eqref{eq:MMV} (i.e. $KQ$ pilots of length $M$, with $Q\geq \Na$) until the table is filled, whereas a two-sided exhaustive search per link would require up to $\Na^{2}K(K-1)/2$ pilots.
On the other hand, by transmitting pilot sequences as described in Section \ref{sssec:Seq_gen}, where each transmitter is assigned a pilot with a different index, receivers can separate pilot signals transmitted simultaneously by multiple transmitters, thereby enabling parallel transmissions. To illustrate this, Fig.~\ref{fig:example_orth_sc} shows an example of the received signal in the frequency domain in a noiseless case where 4 devices are transmitting in orthogonal frequency sets, for $M=1024$, $M_{\rm s}=16$. The receiver can then separate the frequency measurements into the corresponding subsets and estimate its beamspace channel with respect to each of the 4 transmitters independently. It only remains to be studied how to partition the network into transmitters and receivers such that all device pairs are aligned.

\begin{figure}
    \centering
    \includegraphics[width=0.9\linewidth]{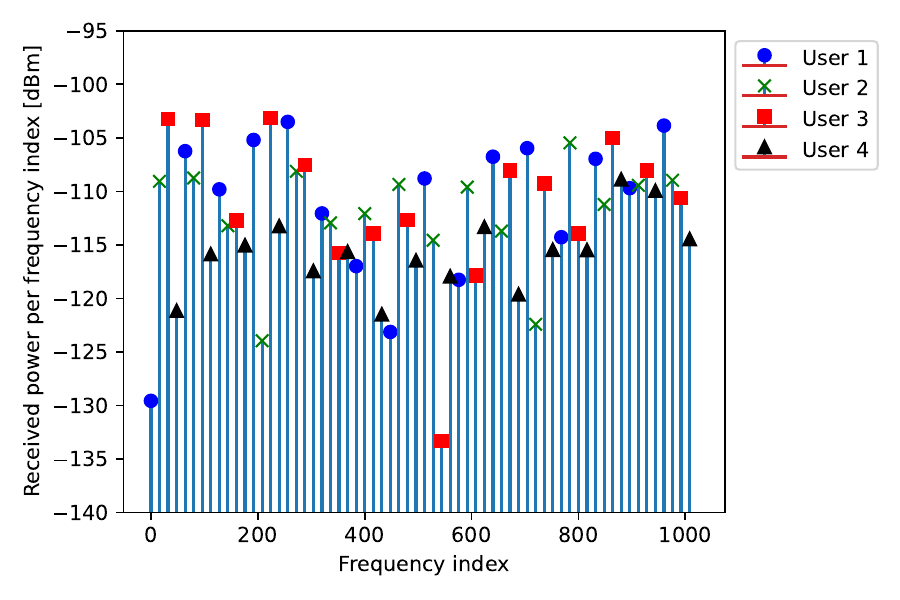}
    \caption{Example of the power distribution per frequency index of the received signal at one receiver in the absence of noise for 4 transmitting devices, each of them with \SI{-15}{\deci\belmilliwatt} transmitted power, when $M=1024$ and $M_{\rm s} = 16$.}
    \label{fig:example_orth_sc}
\end{figure}

For simplicity, we analyze the most unfavourable scenario in which all $K$ devices want to align with each other, resulting in $K(K-1)/2$ total pairs, in order to obtain a pessimistic upper bound on the resources needed for distributed beam alignment. Such a setup can be described by an undirected complete graph $\Gc$ of $K$ nodes. 
We aim to label each of the edges with the beam pair index connecting the corresponding pair of devices. 
Under the mild assumption that $K \leq M/M_{\rm s}$, devices can be partitioned into sets $\Tc_{1}$ and $\Rc_{1}$ of transmitters and receivers respectively, with $\Tc_{1} \cap \Rc_{1} = \emptyset$ and $|\Tc_{1} \cup \Rc_{1}| = K$. Each device $k\in\Tc_{1}$ is assigned a frequency set $\Fc_{k}$ as described in Section \ref{sssec:Seq_gen}. Then, $|\Tc_{1}||\Rc_{1}|$ out of the $K(K-1)/2$ edges can be labeled with a round of measurements of the type \eqref{eq:MMV}. We repeat the procedure with different sets of transmitters and receivers $\Tc_{i}$, $\Rc_{i}$ until all edges are labeled. In the remainder, we refer to every set of transmissions under the same partition of transmitters and receivers as an alignment round.

The problem at hand is thus optimizing the sets of transmitters and receivers at each alignment round as to minimize the total number or rounds needed. While the problem is hard to solve, we can obtain a simple upper bound by considering the greedy approach that tries to label as many edges as possible at each round, which is equivalent to solving a sequence of max-cut problems \cite{CombinatorialOpt}. The max-cut problem is NP-hard in general, but becomes trivial for complete graphs. We start by adding $2^{\lceil\log_{2}K\rceil} - K$ dummy nodes to the graph $\Gc$, such that the number of vertices $K'=2^{\lceil\log_{2}K\rceil}$ of the augmented graph $\Gc'$ is a power of 2. Consider a first partition with transmitters in $\Tc'_{1}$ and receivers in $\Rc'_{1}$, with $|\Tc'_{1}| + |\Rc'_{1}| = K'$. Since $K'$ is even, from the AM-GM inequality \cite{AM_GM},
\begin{align}\label{eq:max_cut_first_round}
    |\Tc'_{1}||\Rc'_{1}| \leq \left(\frac{|\Tc'_{1}| + |\Rc'_{1}|}{2}\right)^{2} = \frac{K'^{2}}{4},
\end{align}
with equality when $|\Tc'_{1}|=|\Rc'_{1}|$. Therefore, the max cut at the first round is trivially achieved by partitioning the $\Gc'$ into two sets of $K'/2$ vertices. 

After removing the labeled edges, the resulting graph is composed of two disconnected subgraphs $\Gc'_{2, 1}$ and $\Gc'_{2, 2}$, each of which is complete and has $K''=K'/2$ vertices. Since $\Gc'_{2, 1} \cap \Gc'_{2, 2} = \emptyset$ (i.e. the set of vertices of $\Gc'_{2, 1}$ and the set of vertices of $\Gc'_{2, 2}$ are non overlapping, and identically for the sets of edges), the max-cut of $\Gc'$ is given by the combination of the max-cut of $\Gc'_{2, 1}$ and the max cut of $\Gc'_{2, 2}$. Since $K'$ is a power of 2, $K''$ is even and we can apply \eqref{eq:max_cut_first_round} to prove that we can label at most $K''^2/4$ edges of each of the subgraphs. Combining the cuts in the two subgraphs, we obtain
\begin{align}
    \frac{(K'')^{2}}{4} + \frac{(K'')^{2}}{4} = \frac{(K'/2)^{2}}{4} + \frac{(K'/2)^{2}}{4} = \frac{K'^{2}}{8}
\end{align}
of the remaining links. Iterating on this idea, at the $i$-th alignment round we will have labeled
\begin{align}
    \sum_{\rho=0}^{i-1}2^{\rho}\frac{(K'/2^{\rho})^{2}}{4} = \sum_{\rho=0}^{i-1}\frac{K'^{2}}{2^{\rho + 2}} = \left(1 - \frac{1}{2^{i}}\right)\frac{K'^{2}}{2}.
\end{align}
The minimum number of alignment rounds to label all the links can be found as
\begin{align}
    i_{\rm opt} &= \min\left\{i\in\ZZ \;\bigg| \left(1 - \frac{1}{2^{i}}\right)\frac{K'^{2}}{2} \geq \frac{K'(K'-1)}{2}\right\}\nonumber\\
    &=\min\left\{i\in\ZZ \;\bigg| 2^i \geq K'\right\}\nonumber\\
    &=\log_{2}K'
\end{align}
since $K'$ is a power of two. Therefore, for the network of $K$ nodes, the number of alignment rounds is upper bounded by $\log_{2}K'=\log_{2}2^{\lceil\log_{2}K\rceil}=\lceil\log_{2}K\rceil$. We can thus reduce the total overhead of network-wide \ac{BA} from the conventional linear scaling with the number of devices to a logarithmic regime.

    \section{Communication phase}\label{sec:CommPhase}
We now describe how to apply the results of the previous sections to enable robust information sharing between devices in the network. The \ac{DFT} codebook we considered for beam alignment was suitable for an estimator based on compressed sensing as it provides a basis on which the channel is sparse. However, its sensitivity to small pointing errors can result in detrimental performance when used over a longer period of time in which devices might undergo macroscopic motion. Therefore, we propose to trade off peak performance with robustness by considering the use of wider beams for communication. In particular, we design a \textit{template} beam pattern with a desired beamwidth and generate a new codebook by steering it towards a finite number of directions. While many designs are readily available from the digital filter theory literature, we considered the Remez exchange algorithm \cite{Remez} for its flexibility and ease of implementation. An example of the resulting design, which we refer to as flat-top beam, for different beamwidths is shown in Fig.~\ref{fig:flattop_beams}. 

\begin{figure}
    \centering
    \includegraphics[width=0.9\linewidth]{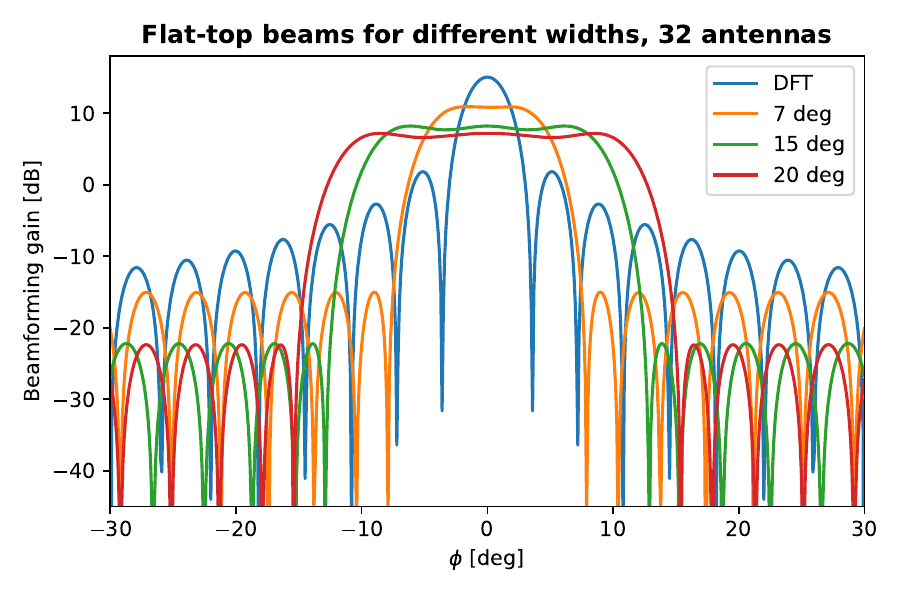}
    \caption{Example of the beams considered for communication as described in Section \ref{sec:CommPhase} when $N_{\rm a} = 32$. In all cases, the \ac{DFT} codebook has the narrower beamwidth on the main lobe and therefore yields the largest beamforming gain. However, it is the most sensitive to misalignment errors and the high sidelobes result in insufficient rejection of signals impinging from undesired directions.}
    \label{fig:flattop_beams}
\end{figure}

Let the coefficient vector generating the desired taper be given by $\omegav$, where $\|\omegav\|^{2} = 1$, and the estimated index in the \ac{DFT} basis of the best beam for a given transmitter to communicate with a receiver be $i_{\rm opt}$. Then, the actual beamforming vector used for communication is given by $\vv_{\rm C} = \diag(\omegav)\cv_{i_{\rm opt}}$. The same methodology is followed at the receiver.

In order to assess the feasibility of single carrier communication without equalization after beam alignment, we investigate now the effect of multipath before and after beamforming in our scenario of interest by means of ray tracing simulations produced by the state of the art software Wireless InSite \cite{InSite}. 
First, we present in Fig.~\ref{fig:freq_response_RT} an example of the frequency response of the channel before and after beamforming for a given placement of two devices, when the beams chosen are the flat-top beams from the codebook maximizing beamforming gain. Clearly, the frequency response becomes approximately flat after beamforming, illustrating that most of the multipath is effectively supressed This result is confirmed in Fig.~\ref{fig:SE_vs_MFB}, which compares the performance of a suboptimal receiver that just treats the first multipath component as useful signal while the rest are treated as interference, and the \ac{MFB}, which is an upper bound attained by an optimal receiver that perfectly equalizes the channel and captures the energy of all multipath components. In this case, a transmitter location was fixed, while a receiver was placed at different locations in a room to evaluate performance against device distance and proximity to walls. The figure shows that single carrier transmission without equalization attains the performance upper bound at most locations. While before beamforming, the gap between both curves is also small, the very low \ac{SNR} makes reliable communication at positive rates infeasible at most locations.

\begin{figure}
    \centering
    \includegraphics[width=\linewidth]{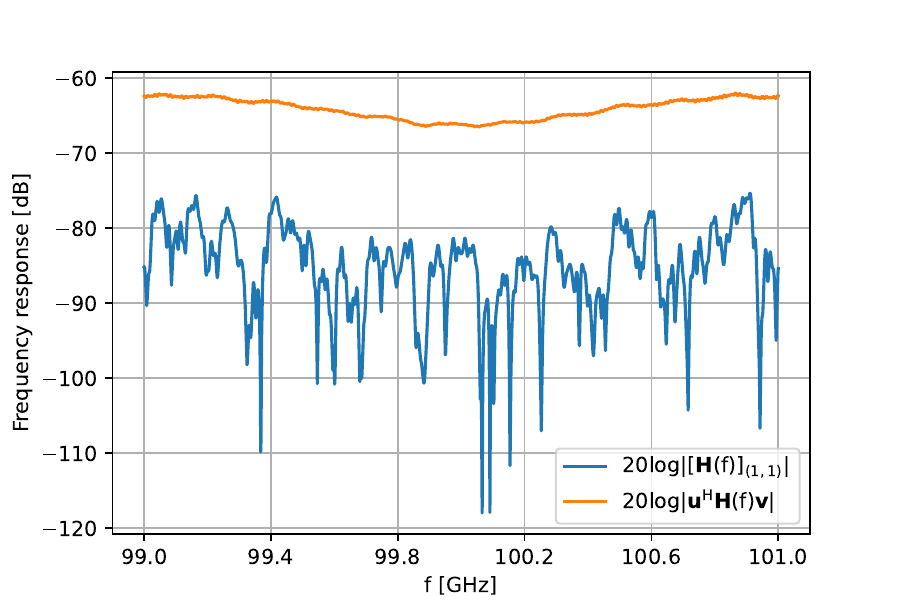}
    \caption{Frequency response of the channel before and after beamforming, when $N_{\rm a} = 32$, the central frequency is \SI{100}{\giga\hertz}, and the beamforming is done with the optimal beam pair of $\ang{7}$ beams as in Fig.~\ref{fig:flattop_beams}.}
    \label{fig:freq_response_RT}
\end{figure}

\begin{figure}
    \centering
    \includegraphics[width=0.8\linewidth]{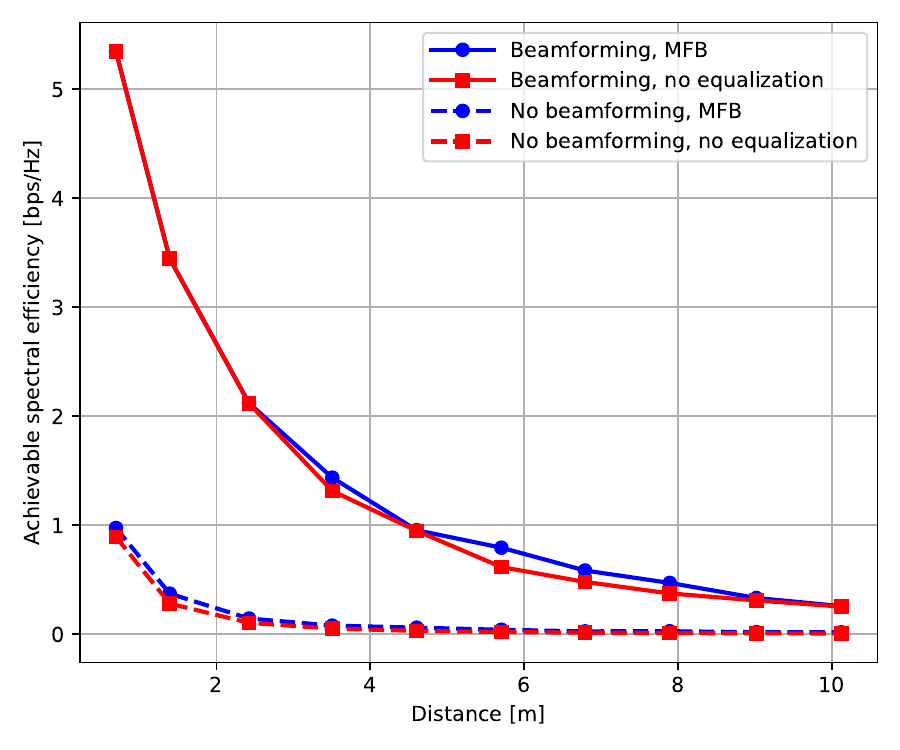}
    \caption{Achievable spectral efficiency before and after beamforming compared with the \ac{MFB} for \SI{-15}{\deci\belmilliwatt} transmit power at \SI{100}{\giga\hertz}. Beamforming is done with 32 antennas and flat-top beams of $\ang{7}$ beamwidth in all cases.}
    \label{fig:SE_vs_MFB}
\end{figure}

    \section{Numerical evaluation}\label{sec:Results}

\subsection{Baseline Scheme}

To the best of our knowledge, the problem of achieving pairwise beam alignment in a \ac{D2D} network has not been studied in the literature, therefore we cannot compare the performance of our method against an existing baseline. Instead, we propose a "802.11ad inspired" baseline where transmitters probe different directions while receivers listen omnidirectionally, thereby estimating a beam index per alignment round instead of a beam index pair. Since we are interested in the case of low $Q$, beam sweeping might not be an option, since there might not be enough time to sweep accross all beams in an alignment round. Therefore, we propose a baseline consisting of transmitter side random beam probing, where each transmitter samples several directions simultaneouly for several rounds. 

For this, we introduce the one-sided beamspace vector
\begin{align}
    \check{\hh}_{\rm os}[m] \coloneqq \frac{1}{\sqrt{\Na}}\onev_{\Na}^\T\Cm^\H\tilde{\Hm}[m]\Cm,
\end{align}
which is the columnwise average of the beamspace matrix $\Cm^\H\tilde{\Hm}[m]\Cm$ at frequency index $m$, and $\check{\HH}_{\rm os} \coloneqq [\check{\hh}^{\T}_{\rm os}[0], \dots, \check{\hh}^{\T}_{\rm os}[M-1]] \in \CC^{\Na\times M}$. A measurement model similar to \eqref{eq:MMV} can be derived, resulting in 
\begin{align}\label{eq:meas_bl}
    \Ym_{\rm os} = \frac{1}{\sqrt{M}}[\check{\vv}_{1},\dots,\check{\vv}_{Q}]^\T\check{\HH}_{\rm os}\diag(\tilde{\sv}) + \Zm_{\rm os},
\end{align}
where $\Zm_{\rm os} \in \CC^{Q \times M}$. In this case, we combine the measured samples at different frequencies noncoherently for each symbol, to resemble conventional energy type measurements. Finally, the channel gain per transmitter beam is estimated at the receiver solving a similar problem to \eqref{eq:cs_mmv}, where the $\ell_{2, 1}$ norm is replaced by the $\ell_{1}$ norm.

Notice that this method can also be put in the framework of Section \ref{sec:NetworkBA} such that multiple devices can simultaneously transmit on orthogonal channels, resulting on a logarithmic scaling of the number of alignment rounds with the number of devices. However, due to the \textit{one-sided} nature of this method, the number of alignment rounds doubles with respect to our proposed approach that jointly estimates the beams for each pair of devices, i.e. the total number of required alignment rounds for the baseline is $2\lceil\log_{2}K\rceil$\footnote{More rigorously, for a sequence of partitions $\{\Tc_{i}, \Rc_{i}\}$, with $i=1,\dots, \lceil\log_{2}K\rceil$ as described in Section~\ref{sec:NetworkBA} for our method, the baseline would adopt a sequence $\{\Tc^{({\rm b})}_{i'}, \Rc^{({\rm b})}_{i'}\}$ with $i' = 1,\dots,2\lceil\log_{2}K\rceil$ where $(\Tc^{({\rm b})}_{2i-1}, \Rc^{({\rm b})}_{2i-1}) = (\Tc_{i}, \Rc_{i})$ and $(\Tc^{({\rm b})}_{2i}, \Rc^{({\rm b})}_{2i}) = (\Rc_{i}, \Tc_{i})$.}.
Furthermore, while our method has the potential of reconstructing general sparse beamspace matrices, the baseline approach is only able to produce rank-1 approximations obtained as the outer product of the one dimensional estimates obtained at each end of a link. An illustrative example is shown in Fig.~\ref{fig:beamspace_ex}. The figure clearly shows that, while the baseline successfully recovers the position of the strongest beam pair, it fails to estimate secondary paths, due to the independent estimation of the transmit and receive beamspace vectors. This may complicate subsequent operations such as scheduling or fast reconfiguration due to blockage of the main path. On the other hand, our method is able to recover some of the strong \ac{NLOS} components.

\begin{figure*}
    \centering
    \subfloat[Ground truth beamspace matrix]{\includegraphics[width=0.3\linewidth]{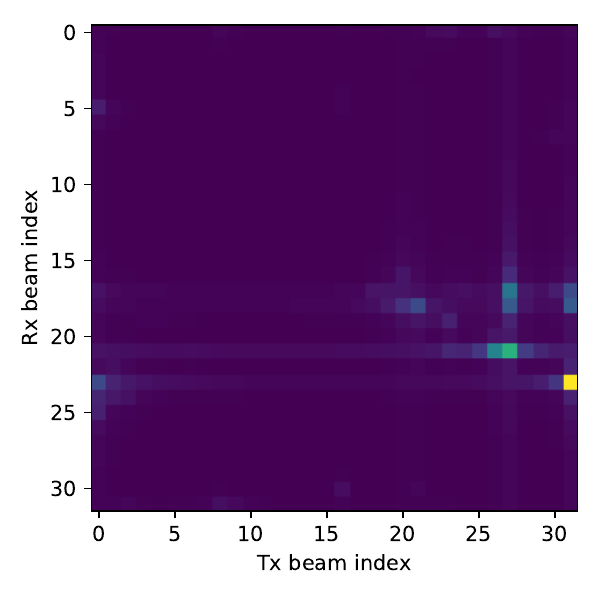}}
    \qquad
    \subfloat[MMV estimate]{\includegraphics[width=0.3\linewidth]{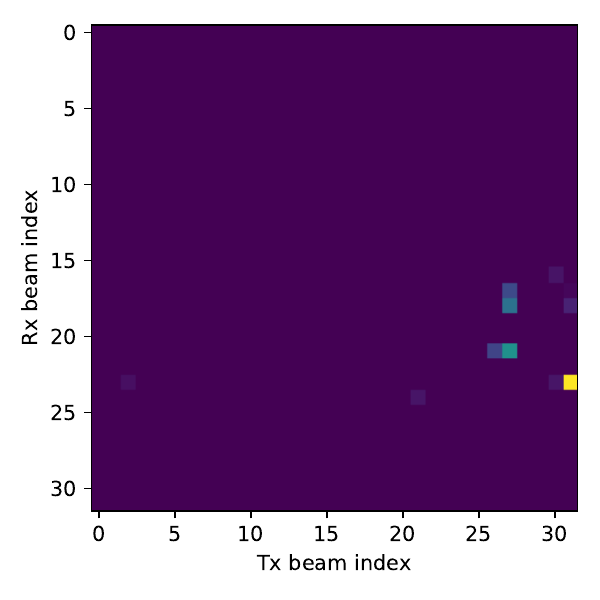}}
    \qquad
    \subfloat[Baseline estimate]{\includegraphics[width=0.3\linewidth]{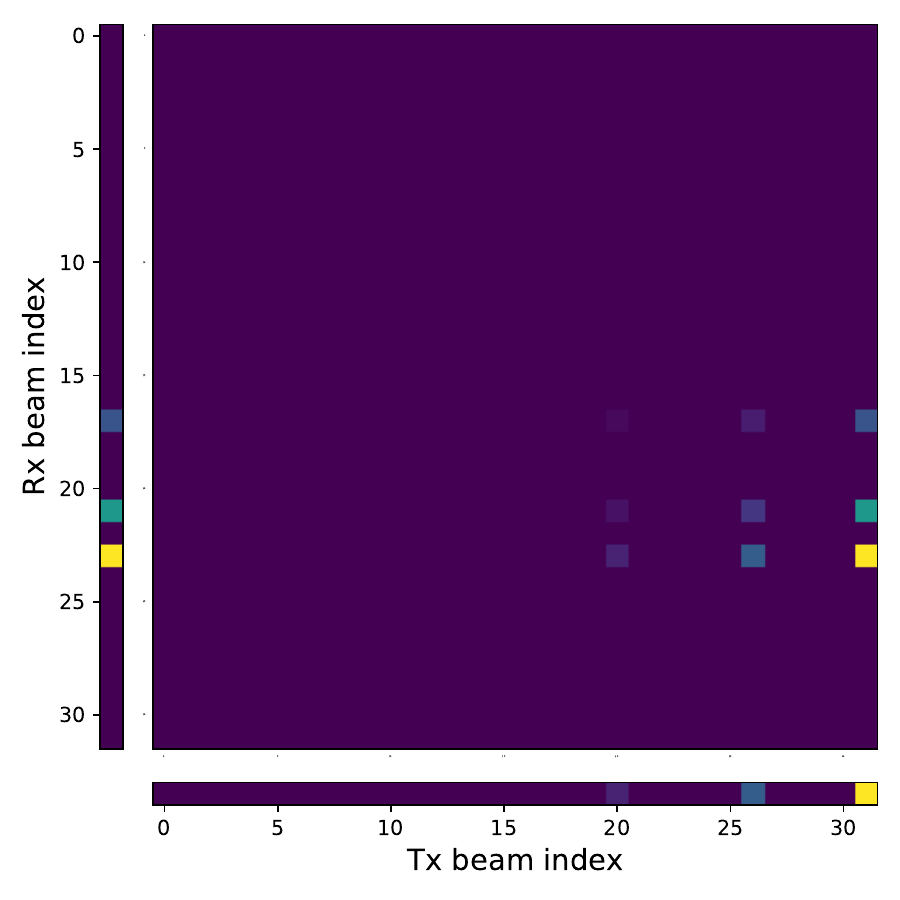}}
    \caption{Illustration of the different results provided by our method and the considered baseline. (a) shows the ground truth beamspace matrix to be estimated. (b) illustrates the result of our method, providing a sparse reconstruction of the beamspace matrix. (c) illustrates the baseline method. First, the device labeled as Tx performs a beam sweep while Rx listens omnidirectionally, such that Rx obtains the represented row vector. Subsequently, the roles are reversed, and Tx now obtains the column vector based on a beam sweep at Rx. A rank-1 approximation to the beamspace matrix can be obtained as an outer product of both measurements. }
    \label{fig:beamspace_ex}
\end{figure*}

\subsection{Results}

\begin{table}
    \centering
    \begin{tabular}{|c|c|}
        \hline
        $f_{0} = $ \SI{100}{\giga\hertz} & $B = $ \SI{2}{\giga\hertz}\\
        \hline
        $M = $ 1024 & $N_{a} = $ 32\\
        \hline
        $K  = $ 8 & $N_{\rm 0}$ = \SI{-174}{\deci\belmilliwatt/\hertz}\\
        \hline
        \multicolumn{2}{|c|}{Communication beamwidth $\approx \ang{7}$}\\
        \hline
    \end{tabular}
    \caption{System parameters.}
    \label{tab:sys_params}
\end{table}

\begin{figure}
    \centering
    \includegraphics[width=0.75\linewidth]{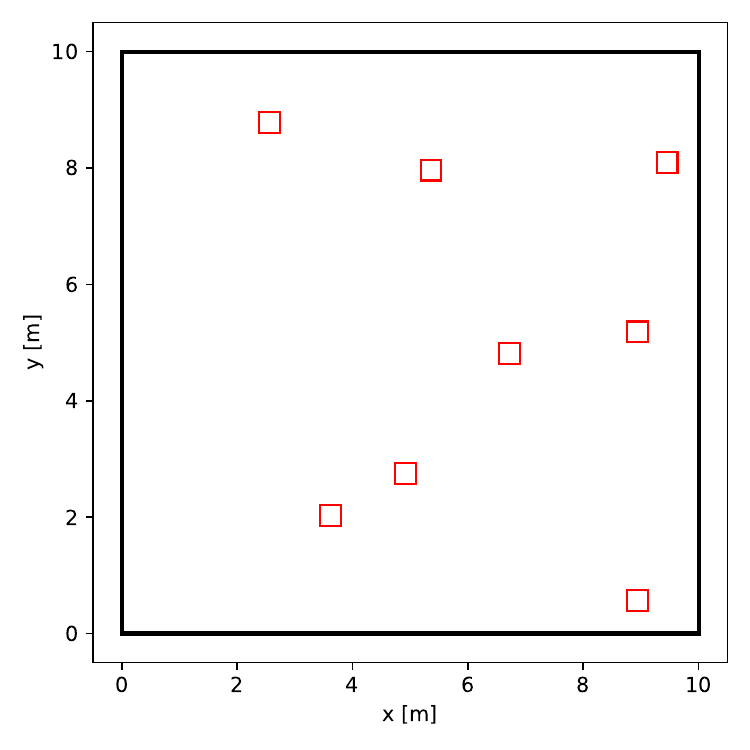}
    \caption{Position of the devices for the performance evaluation of network-wide beam alignment.}
    \label{fig:DevPosAll}
\end{figure}

\begin{figure}
    \centering
    \subfloat[\SI{-20}{\deci\belmilliwatt}]{\includegraphics[width=0.9\linewidth]{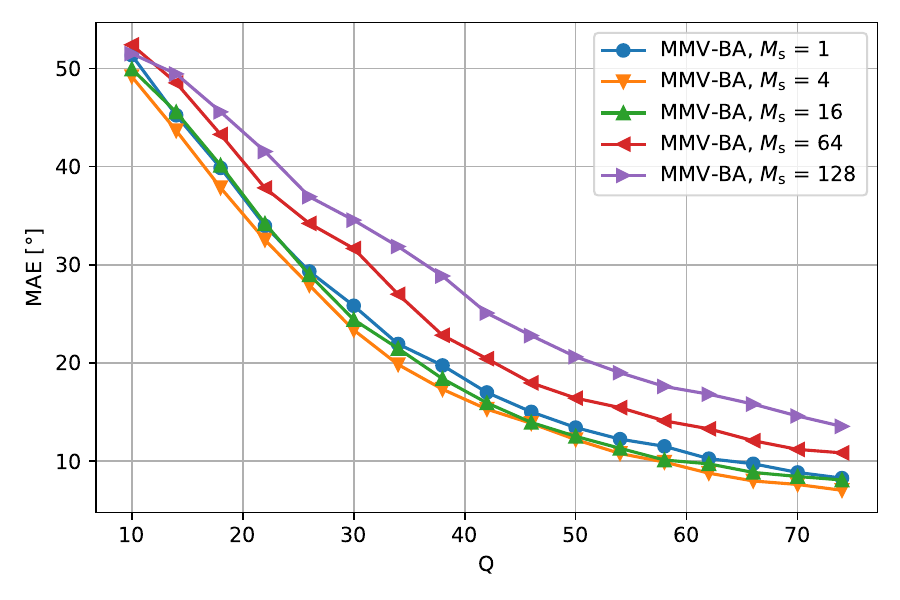}}
    \\
    \subfloat[\SI{-15}{\deci\belmilliwatt}]{\includegraphics[width=0.9\linewidth]{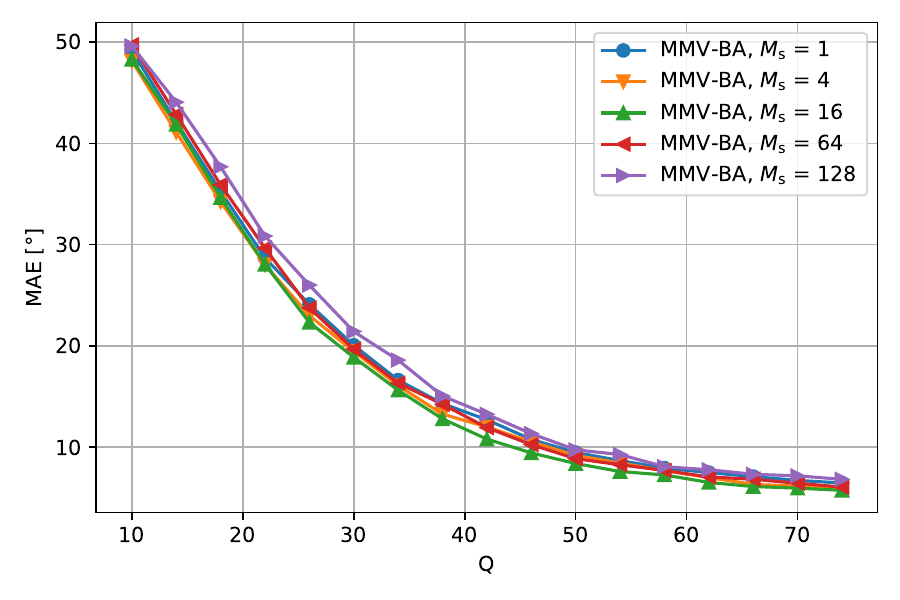}}
    \caption{Mean absolute error of angle estimation as a function of the number of pilots per alignment round considering all pairs of devices in the network considered.}
    \label{fig:AngleMAE}
\end{figure}

\begin{figure*}
    \centering
    \subfloat[\SI{-20}{\deci\belmilliwatt}, $Q=30$.]{\includegraphics[width=0.3\linewidth]{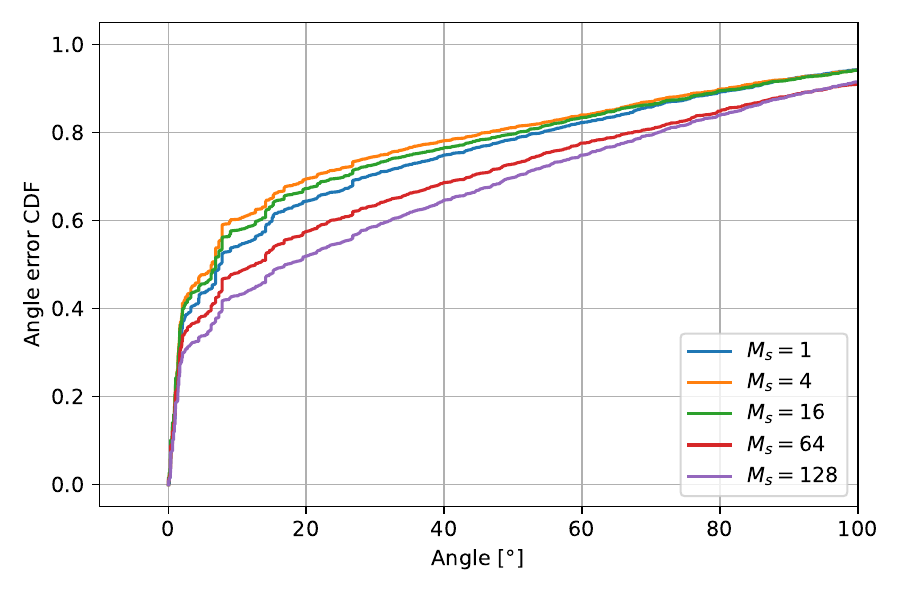}}
    \qquad
    \subfloat[\SI{-20}{\deci\belmilliwatt}, $Q=50$.]{\includegraphics[width=0.3\linewidth]{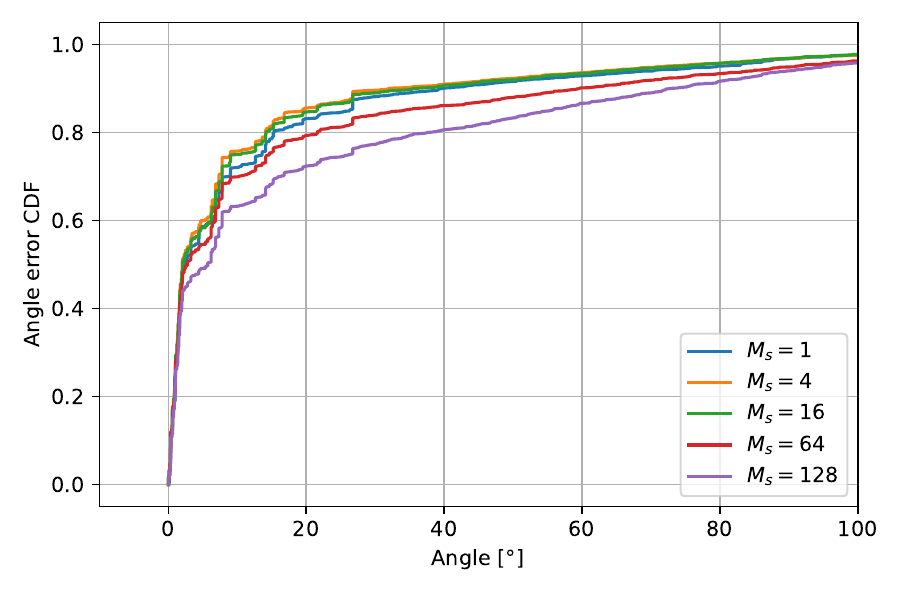}}
    \qquad
    \subfloat[\SI{-20}{\deci\belmilliwatt}, $Q=70$.]{\includegraphics[width=0.3\linewidth]{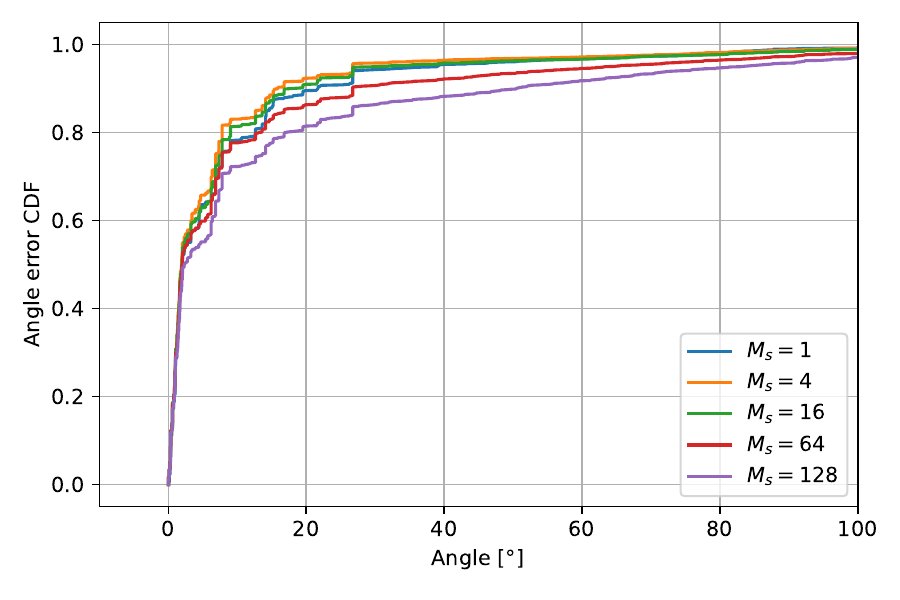}}
    \\
    \subfloat[\SI{-15}{\deci\belmilliwatt}, $Q=30$.]{\includegraphics[width=0.3\linewidth]{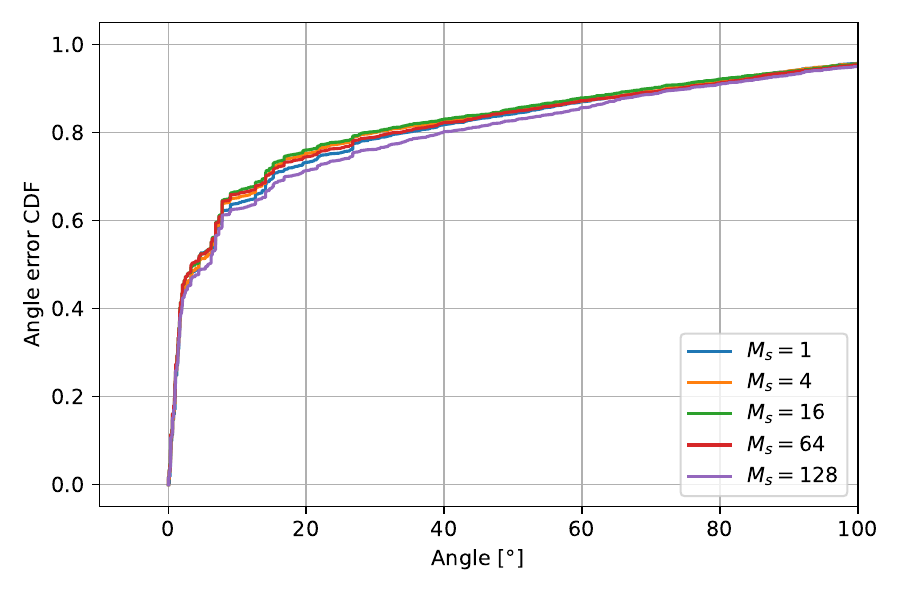}}
    \qquad
    \subfloat[\SI{-15}{\deci\belmilliwatt}, $Q=50$.]{\includegraphics[width=0.3\linewidth]{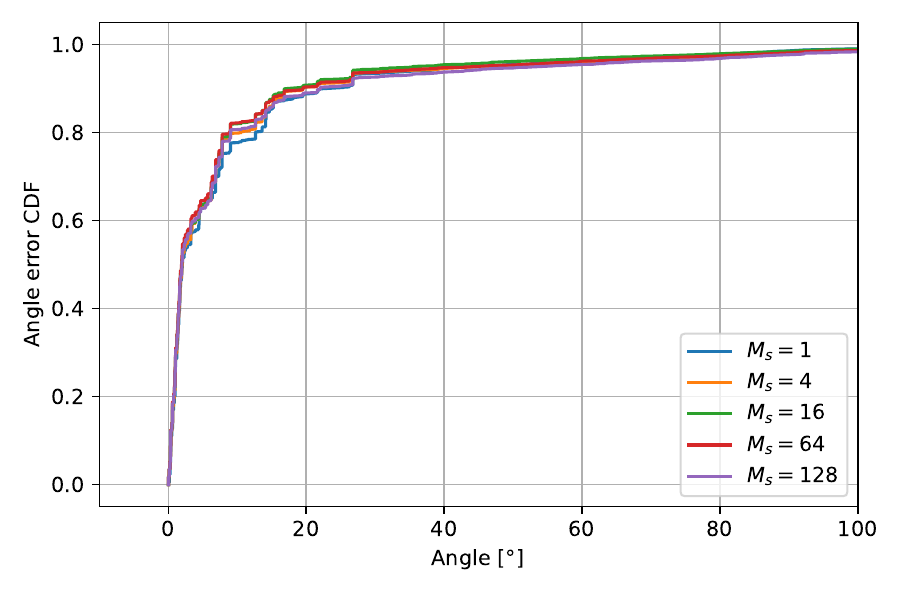}}
    \qquad
    \subfloat[\SI{-15}{\deci\belmilliwatt}, $Q=70$.]{\includegraphics[width=0.3\linewidth]{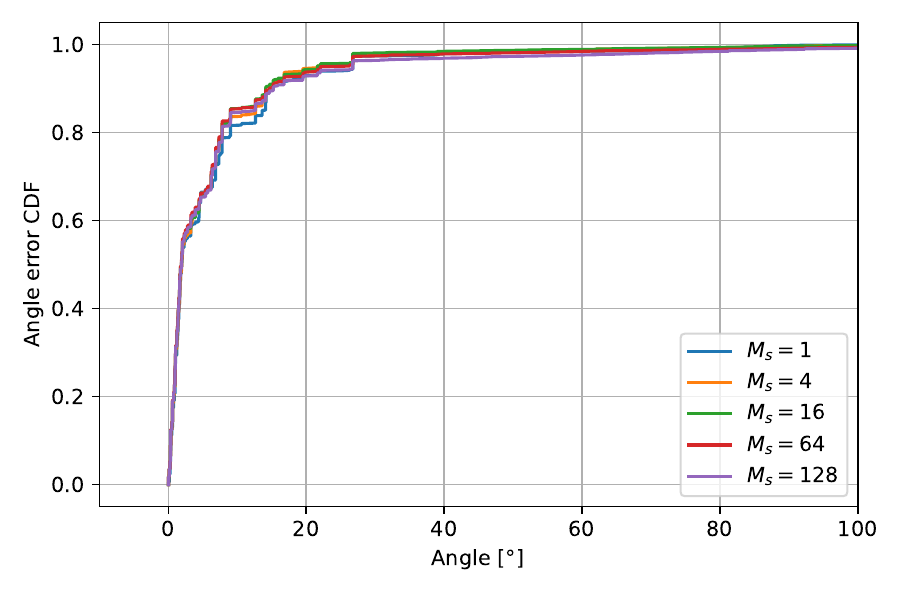}}
    \caption{CDF of angle error for different transmit power and pilot symbols per alignment round.}
    \label{fig:AngleCDF}
\end{figure*}

We evalute the methods described here in a \SI{10}{\meter} $\times$ \SI{10}{\meter} $\times$ \SI{3}{\meter} room with $K=8$ devices as illustrated in Fig.~\ref{fig:DevPosAll}. All relevant system parameters considered are specified in Table~\ref{tab:sys_params}.
In particular, we run full rounds of network-wide alignment such that at the end of each run we obtain an estimate of the table illustrated in Fig.~\ref{fig:D2D_net_and_table}. For each simulation and each alignment round, we sample a new sampling matrix in order to average the effect of each specific realization over the ensemble of feasible matrices we presented. 

We start by examining the performance of beam alignment in terms of the angular error between the ground truth angle of arrival of the \ac{LOS} path between each pair of devices and the central angle of the beam estimated for that pair. Fig.~\ref{fig:AngleMAE} shows the \ac{MAE} as a function of the total number of pilots per alignment round for two different values of the transmit power. It can be clearly observed that when the power is more limited, spreading energy across a large number of frequencies can have a very detrimental effect on performance. Furthermore, it can be observed that the error does not evolve monotonically as a function of the number of frequencies. Instead, we see better performance at intermediate spreading levels. This result confirms the value of our pilot design, which allows to flexibly decide on the amount of spreading in frequency domain. On the other hand, for higher values of transmit power, spreading does not seem to have a substantial effect. A more detailed statistical view on the error performance is shown in Fig.~\ref{fig:AngleCDF}, which shows the \ac{CDF} of the angle error for certain fixed values of the number of pilots per alignment round $Q$. Similar to the \ac{MAE} results, we observe a higher dependency on the frequency spreading for the lower power configuration. It can be observed that the slope of the \ac{CDF} curves is high until around $\ang{7}$. Therefore, for the following results, we design a flat-top beam of $\ang{7}$ beamwidth.

To present results in terms of achievable sum rates, we group devices into pairs as illustrated in Fig.~\ref{fig:DevicePairs}, where devices transmit and receive applying the flat-top beams indicated by the estimated alignment table. Let us now refer with $\vv^{(\psf)}$ (respectively $\uv^{(\psf)}$) to the beamforming vector chosen by the transmiter (respectively receiver) of the $\psf$-th beam pairs indicated in Fig.~\ref{fig:DevicePairs}, with $\psf\in\{1, 2, 3, 4\}$. 
Referring now to the discrete-time \ac{MIMO} channel between the transmitter of the $\psf_{1}$-th pair and the receiver of the $\psf_{2}$-th pair as $\overline{\Hm}_{\psf_{1}, \psf_{2}}[n]$ (similar to \eqref{eq:DT_IO}), we compute the achievable sum spectral efficiency as 
\begin{align}\label{eq:se}
    {\rm SE} = \sum_{\psf=1}^{4}\log_{2}\left(1 + {\rm SINR}_{\psf}\right),
\end{align}
where
\begin{align}\label{eq:SINR}
    {\rm SINR}_{\psf} = \frac{P_{\rm tx}|(\uv^{(\psf)})^\H\overline{\Hm}_{\psf, \psf}[0]\vv^{(\psf)}|^{2}}{N_{\rm 0}B + {\rm ISI}(\psf) + \sum_{\psf'\neq \psf}{\rm MUI}(\psf, \psf')},
\end{align}
where ${\rm ISI}(\psf)$ is the intersymbol interference assuming no equalization, which is given by
\begin{align}
    {\rm ISI}(\psf) \coloneqq P_{\rm tx}\sum_{n=1}^{L}|(\uv^{(\psf)})^\H\overline{\Hm}_{\psf, \psf}[n]\vv^{(\psf)}|^{2}
\end{align}
and ${\rm MUI}(\psf, \psf')$ is the multi-user interference at the receiver of the $\psf$-th pair caused by the transmitter of the $\psf'$-th pair, which is given by
\begin{align}
    {\rm MUI}(\psf, \psf') \coloneqq P_{\rm tx}\sum_{n=0}^{L}|(\uv^{(\psf)})^\H\overline{\Hm}_{\psf', \psf}[n]\vv^{(\psf')}|^{2}.
\end{align}
In \eqref{eq:SINR}, $\overline{\Hm}_{\psf, \psf}[0]$ is the channel matrix associated with the \ac{LOS} path of the $\psf$-th pair, while $\overline{\Hm}_{\psf, \psf}[n]$, for $n>0$ is treated as \ac{ISI}.

\begin{figure}
    \centering
    \includegraphics[width=0.8\linewidth]{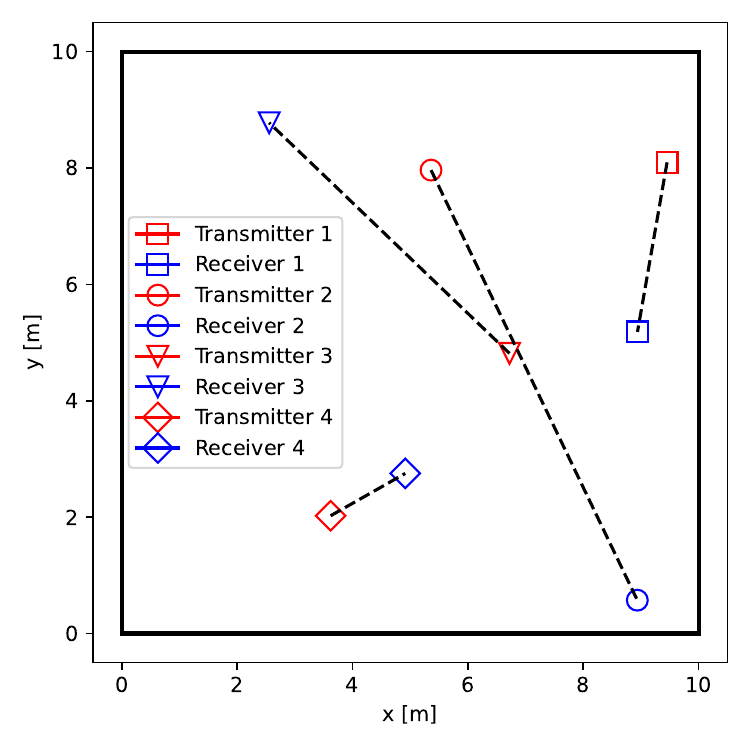}
    \caption{Device pairing considered to evaluate the achievable sum spectral efficiency.}
    \label{fig:DevicePairs}
\end{figure}

We present the performance improvement of our method with respect to the baseline by fixing the transmit power to $\Ptx=$\SI{-20}{\deci\belmilliwatt} and \SI{-15}{\deci\belmilliwatt}, and evaluating the achievable results as a function of the total budget of pilot symbols\footnote{Notice that after our network-wide \ac{BA} procedure, all devices know the beam configuration to communicate with any other node in the network. We only group them in pairs for the sake of presenting results in terms of achievable sum rates.}. Fig.~\ref{fig:sr_vs_Q_m20dBm} and Fig.~\ref{fig:sr_vs_Q_m15dBm} illustrate the achievable sum spectral efficiency after network-wide beam alignment, which clearly show much faster convergence of our method to a good alignment condition than the baseline. It is observed that in the lower \ac{SNR} case, our method achieves best performance when the energy is distributed among 4 frequencies, whereas the worst results are obtained for $M_{\rm s}=1$ and $M_{\rm s}=64$. This clearly illustrates the trade-off we described in the text and in the angle estimation results, also observed in other works such as \cite{Xiaoshen2018scalable}. When the pilot is spread across many indices, each mesurement is too noisy to produce accurate estimates, whereas if the pilot is too concentrated in frequency, the number of signal dimensions for estimation becomes too low. Converesly, in the \SI{-15}{\deci\belmilliwatt} case, all levels of spreading achieve very similar performance. This points to the fact that the \ac{SNR} is high enough to allow for frequency spreading, but not much information is gained from every frequency measurement due to the high correlation of the channel response in frequency domain (see Fig.~\ref{fig:freq_response_RT}). For the baseline, we observe that better performance is generally obtained with lower spreading, which can be explained by the noncoherent combining of the samples accross frequencies we considered.

\begin{figure}
    \centering
    \includegraphics[width=0.9\linewidth]{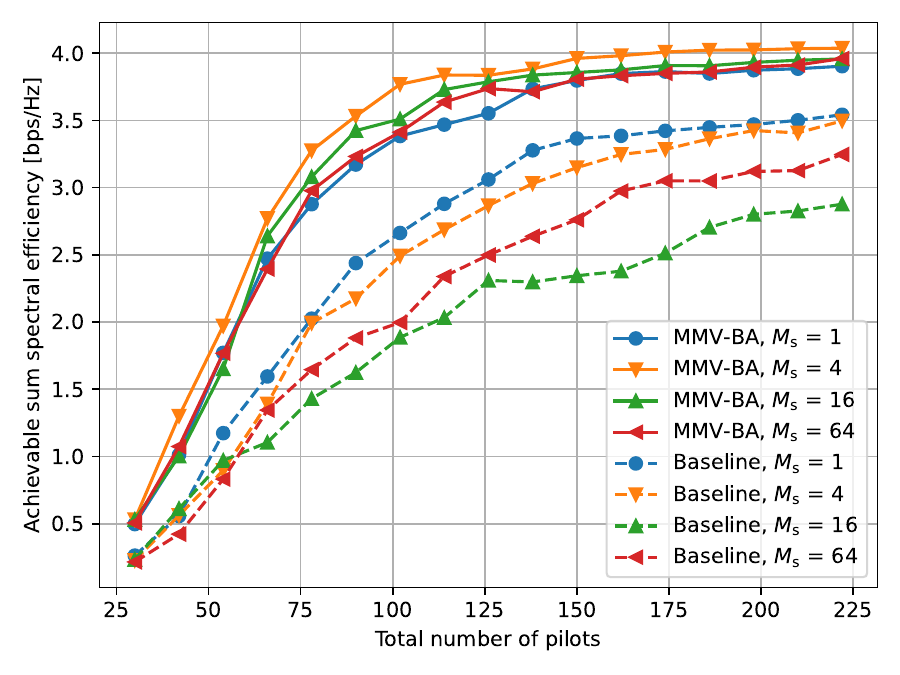}
    \caption{Achievable sum spectral efficiency after network-wide beam alignment as a function of the total number of pilot sequences when the transmit power is fixed to \SI{-20}{\deci\belmilliwatt}.}
    \label{fig:sr_vs_Q_m20dBm}
\end{figure}

\begin{figure}
    \centering
    \includegraphics[width=0.9\linewidth]{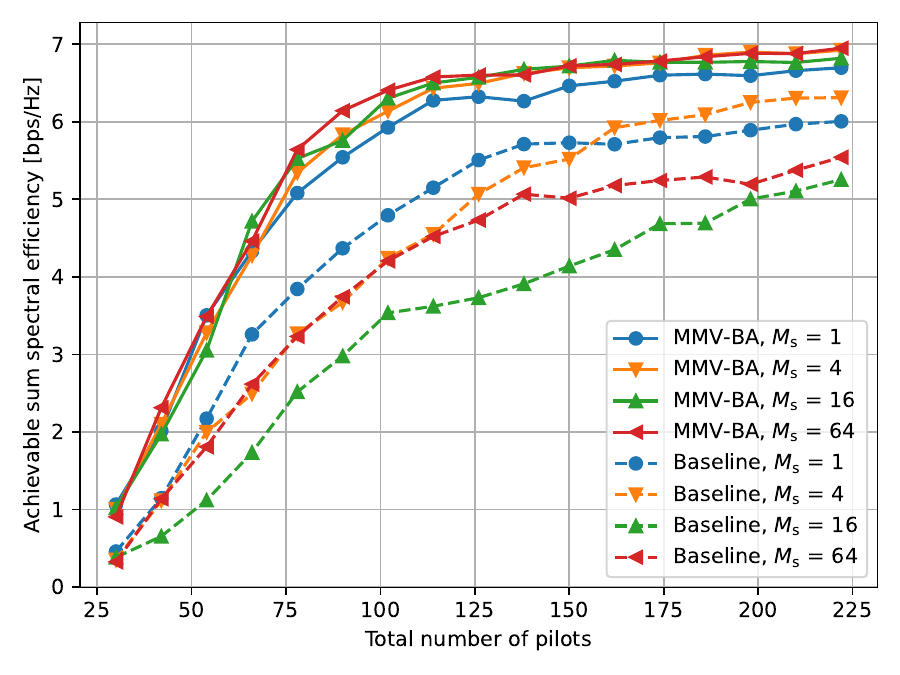}
    \caption{Achievable sum spectral efficiency after network-wide beam alignment as a function of the total number of pilot sequences when the transmit power is fixed to \SI{-15}{\deci\belmilliwatt}.}
    \label{fig:sr_vs_Q_m15dBm}
\end{figure}

    \section{Conclusions and Future Work}\label{sec:Conclusions}
In this paper, we described a method for fast network-wide \ac{BA} in \ac{D2D} sub-THz networks. By considering synchronous devices, our method processes signals coherently resulting in a large reduction in the time required for all devices to be aligned with each other. Furthermore, we proposed a simple procedure to design pilot sequences that concentrates energy in only a subset of frequencies while satisfying a constant envelope constraint in the time domain. Finally, we propose a strategy to apply the described methodology to network scenarios and quickly achieve network-wide alignment. Our methods are verified for different power conditions in a room without obstacles by means of ray tracing simulations. Besides, we show experimentally that at sub-THz frequencies, single carrier communication without equalization is almost optimal after beamforming.

    \appendix
    \subsection{Design of i.i.d. sampling matrices}\label{ap:SamplingMatrix}
We describe in this appendix how a specific design of the beam weighting vectors introduced in Section \ref{sec:SingleBA} results in an i.i.d. sampling matrix, such that the theory and algorithms of \ac{CS} can be directly applied. The design is motivated by the following Lemma.
\begin{lemma}\label{lem:kron_iid_uc}
    Denote as $(\uc, \boreluc)$ the complex unit circle with the Borel $\sigma$-algebra restricted to the unit circle. As a natural measure of lengths on $\uc$, consider the push-forward $\mu=\lambda\circ f^{-1}$ of the Lebesgue measure $\lambda$ on $[0,2\pi)$ under $f:[0,2\pi)\rightarrow\uc, t\mapsto e^{\jim t}$.
    Let $X_i$ be i.i.d. random variables uniformly distributed on $(\uc, \boreluc, \mu)$, i.e. $\PP(X_i\in A)=\frac{\mu(A)}{2\pi}$ for all $A\in\boreluc$, for $i=1,2,3$. The random variables $Y=X_{1}\cdot X_{2}$ and $Z=X_{1}\cdot X_{3}$ are as well independent and uniformly distributed on $(\uc, \boreluc, \mu)$. 
\end{lemma}
\begin{proof}
    To see that $Y$ and similarly $Z$ is again uniformly distributed on $(\uc, \boreluc, \mu)$, let $A\in\boreluc$. By definition of the pushforward, Tonelli's theorem and the translation-invariance of the Lebesgue measure we have
    \begin{align}
        \PP(Y\in A) & = \PP(X_1 \cdot X_2 \in A)\nonumber\\
        &= \frac{1}{(2\pi)^2}\int_{[0,2\pi)} \PP(X_2\in e^{-\jim t} A) \dif\lambda(t)\nonumber\\
        &= \frac{1}{(2\pi)^2}\int_{[0,2\pi)} \lambda(\{s\in[0,2\pi) | e^{\jim(s+t)} \in A\}) \dif\lambda(t)\nonumber\\
        &= \frac{1}{(2\pi)^2}\int_{[0,2\pi)} \lambda(\{s\in[0,2\pi) | e^{\jim s} \in A\}) \dif\lambda(t)\nonumber\\
        &= \frac{\mu(A)}{2\pi}.
    \end{align}
    With similar arguments and using the independence of $X_2, X_3$, for $A,B\in\boreluc$ it holds that
    \begin{align}
        &\PP(Y\in A, Z\in B) \nonumber\\
        & = \frac{1}{2\pi} \int_{[0,2\pi)} \PP(X_2\in e^{-\jim t} A, X_3\in e^{-\jim t} B) \dif\lambda(t) \nonumber\\
        & = \frac{1}{2\pi} \int_{[0,2\pi)} \PP(X_2\in e^{-\jim t} A)\, \PP(X_3\in e^{-\jim t} B) \dif\lambda(t) \nonumber\\
        & = \frac{1}{2\pi} \int_{[0,2\pi)} \PP(X_2\in A)\, \PP(X_3\in B) \dif\lambda(t) \nonumber\\
        & =  \PP(X_2\in A)\, \PP(X_3\in B) \nonumber\\
        & = \PP(Y\in A)\, \PP(Z\in B),
    \end{align}
    where in the last step we used that $Y,Z,X_2,X_3$ all follow the same distribution.
\end{proof}

We can apply Lemma \ref{lem:kron_iid_uc} to generate i.i.d. sampling matrices respecting the structure defined in \eqref{eq:MMV}. To see this, notice that every entry of $\Na\check{\wv}_{q}\coloneqq \sqrt{\Na}\check{\vv}_{q}\otimes\sqrt{\Na}\check{\uv}_{q}$ is given by $\Na\check{\wv}_{q}^{(ij)}=[\sqrt{\Na}\check{\vv}_{q}]_{i}\cdot [\sqrt{\Na}\check{\uv}_{q}]_{j}$ for a certain combination of $i$ and $j$, where the $\sqrt{\Na}$ and $\Na$ terms are needed for normalization purposes. If we generate each entry of $\sqrt{\Na}\check{\uv}_{q}$ and $\sqrt{\Na}\check{\vv}_{q}$ $\forall q \in [Q]$ by sampling i.i.d. from the distribution described in Lemma \ref{lem:kron_iid_uc}, each entry of $\Na\check{\wv}_{q}$ has been shown to follow the same distribution. Entries $\check{\wv}_{q}^{(i_{1}j_{1})}$ and $\check{\wv}_{q}^{(i_{2}j_{2})}$ are clearly independent when $i_{1}\neq i_{2}$ and $j_{1}\neq j_{2}$ due to the pairwise independence of $[\check{\vv}_{q}]_{i_{1}}$, $[\check{\vv}_{q}]_{i_{2}}$, $[\check{\uv}_{q}]_{j_{1}}$ and $[\check{\uv}_{q}]_{j_{2}}$. If $i_{1}=i_{2}$ or $j_{1}=j_{2}$, Lemma \ref{lem:kron_iid_uc} guarantees independence. Finally, entries of $\check{\wv}_{q}$ and $\check{\wv}_{q'}$ for $q\neq q'$ are independent due to the pairwise independence of $\check{\vv}_{q}$, $\check{\vv}_{q'}$, $\check{\uv}_{q}$ and $\check{\uv}_{q'}$. 

We therefore consider a design of the sampling matrix $\Am$ where each entry is sampled independently and uniformly from the complex unit circle. Formally, we construct $\Am$ following
\begin{equation}
    \begin{split}
       \PP([\Am]_{i, j}\in A) = \mu(A) \qquad {\rm for}\; i\in[Q],\,j\in[N_{\rm a}^{2}]\\
       \text{ and } A \in \boreluc.
    \end{split}
\end{equation}

    \bibliographystyle{IEEEtran}
    \bibliography{IEEEabrv,bibliography}

\end{document}